\documentclass[conference]{IEEEtran}

\usepackage{algorithm}
\usepackage{algpseudocode}
\usepackage{mdframed}
\usepackage{upgreek}
\usepackage{booktabs}
\usepackage{amsmath}
\usepackage{amssymb}
\usepackage{amsthm}
\usepackage{graphicx}
\usepackage{xcolor}
\usepackage{url}
\usepackage{cite}
\usepackage[hidelinks]{hyperref}

\usepackage{enumitem, amsmath, amssymb}
\usepackage{tikz}
\usepackage{varwidth}
\usepackage[operators,sets]{cryptocode}
\usepackage[breakable]{tcolorbox}
\usepackage{cuted}
\usepackage{subcaption}

\renewcommand{\textcolor}[2]{\bgroup\color{#1}#2\egroup}

\newcommand{\update}[1]{\textcolor{black}{#1}}

\newcounter{algo}
\renewcommand{\thealgo}{\arabic{algo}}

\newcommand{\algotitle}[2]{\refstepcounter{algo}\label{#1}Alg.~\thealgo: #2}
\newcommand{\Note}[1]{{\footnotesize\textcolor{blue}{$\rhd$~#1}}}
\newcommand{\NoteNL}[1]{\\{\footnotesize\textcolor{blue}{$\rhd$~#1}}}
\newcommand{\Op}[1]{\ensuremath{\mathrm{#1}}}
\newcommand{\Proc}[1]{\textsc{#1}}

\newcommand{\sysname}{\text{Trinq\ensuremath{\epsilon}t}}

\newcommand{\Pa}{P}
\newcommand{\Ser}{S}
\newcommand{\Ts}{T_s}
\newcommand{\T}{\mathrm{T}}
\newcommand{\Tu}{T_u}
\newcommand{\To}{T_{\text{Outer}}}

\newcommand{\Seen}{L_s}

\newcommand{\Q}{Q}
\newcommand{\W}{W}
\newcommand{\US}{L_u}
\newcommand{\Rs}{R_{s}}
\newcommand{\Ru}{R_{u}}
\newcommand{\Li}{\mathcal{L}}
\newcommand{\C}{C}
\newcommand{\F}{F}

\newcommand{\B}{\mathcal{B}}
\newcommand{\cnt}{\Op{cnt}}

\newcommand{\rec}{\Op{Rec}}

\newcommand{\Cmp}{\Op{Cmp}}
\newcommand{\Mult}{\Op{Mult}}
\newcommand{\Trunc}{\Op{Trunc}}
\newcommand{\Equal}{\Op{Equal}}
\newcommand{\la}{\langle}
\newcommand{\N}{l}
\newcommand{\calA}{\mathcal{A}}
\newcommand{\ra}{\rangle}

\newcommand{\Br}{\ensuremath{\Op{Br}}}
\newcommand{\Ex}{\ensuremath{\Op{Ex}}}
\newcommand{\K}{\ensuremath{\Op{K}}}
\newcommand{\randarrow}{\xleftarrow{\$}}
\newcommand{\Ds}{\Op{D}}

\newcommand{\G}{\mathcal{G}}

\newcommand{\Fshuffle}{\mathcal{F}_{\Op{Shuffle}}}

\algnewcommand{\Input}{\State \textbf{Input:} }
\algnewcommand{\Output}{\State \textbf{Output:} }

\newenvironment{packeditemize}{
\begin{itemize}[leftmargin=*]
  \setlength{\itemsep}{0pt}
  \setlength{\parskip}{0pt}
  \setlength{\parsep}{0pt}
}{\end{itemize}}

\newenvironment{mybox2}[1]{%
    \begin{tcolorbox}[title={#1}, 
    colback=white, breakable]%
    }{
    \end{tcolorbox}
}

\newtheorem{definition}{Definition}

\newtheorem{theorem}{Theorem}

\begin{document}

\title{\sysname: Private Triangle and Quadrangle Counting over Distributed Graphs}

\author{
\IEEEauthorblockN{
Mushtari Sadia$^{*}$ \qquad Amrita Roy Chowdhury$^{*}$
}
\IEEEauthorblockA{
$^{*}$University of Michigan, Ann Arbor
}
}
\IEEEoverridecommandlockouts
\makeatletter\def\@IEEEpubidpullup{6.5\baselineskip}\makeatother

\maketitle

\begin{abstract}
Triangle and quadrangle counts are core statistics in graph analysis. Yet many real-world graphs are distributed across multiple parties and encode highly sensitive relationships, precluding direct data sharing. Secure multi-party computation (MPC) provides a principled alternative, but introduces a fundamental tension: the protocol must hide the graph's topology, forcing fully data-oblivious operations over \emph{all} potential edges--while real-world graphs are extremely sparse, making most of this work wasted. We propose \sysname, a system that resolves this tension through new MPC-friendly algorithms for private triangle and quadrangle detection. \sysname{} safely exploits sparsity using a suite of novel techniques that eliminate vast numbers of unnecessary operations while preserving strong security in the malicious threat model. \sysname~supports both counting and enumeration and, in our evaluation, outperforms all five baselines by up to $10^{5}\times$ in runtime.
\end{abstract}

\IEEEpeerreviewmaketitle

\section{Introduction}
\label{sec:intro}


Graphs are a fundamental tool for modeling real world relationships, but in many settings no single entity holds the entire graph. Instead, the graph is \textit{distributed} across mutually distrustful parties, each observing only its local subgraph. This setting arises in deployments ranging from decentralized social platforms such as Mastodon, Diaspora, or PeerTube, where each node is a party and edges encode `follower' relationships, to settings where a small set of organizations hold disjoint slices of a joint network, such as two internet service providers holding different portions of a shared topology.

Even for distributed graphs, global statistics such as triangle and quadrangle detection remain fundamental to understanding graph structure~\cite{triangle1,triangle2,triangle3,triangle4,triangle5,triangle6,triangle7,triangle9,cycle1,cycle2,cycle3,cycle4,cycle5,cycle6,chiba1985arboricity}. Triangles (3-cycles) underlie the clustering coefficient, a core measure of social cohesion used in community detection, link prediction, and friend recommendation~\cite{community-detection,friend-recommendation,traffic-flow}. Quadrangles (4-cycles) are similarly important in bipartite settings, including dating platforms, mentor--student networks, and user--item recommenders~\cite{lind2005cycles,talaga2022structural,cycle7,cycle8}. More broadly, both cycles support fraud and collusion detection and the discovery of shared-interest communities~\cite{lu2023graph,becchetti2008efficient,klymko2014using}, making them key building blocks of graph analytics. 

However, directly pooling the parties' graph data is infeasible because edges encode sensitive relationships. In social networks, edges reveal private connections; in ISP networks, topology is proprietary, and its exposure could enable attacks on critical infrastructure. Thus, no party is willing (or legally able) to reveal its raw graph. Secure multi-party computation (MPC) offers a principled alternative, enabling computation over the union of the graphs without exposing any party's inputs. Yet applying generic MPC to graph analytics is notoriously expensive. Naive approaches based on an adjacency-matrix ($A$) representation quickly become prohibitive at scale. For triangle counting, the classic baseline computes $A^3$ (equivalently, $\mathrm{tr}(A^3)/6$), implicitly considering all $\binom{n}{3}$ vertex triples and incurring cubic work and communication. For quadrangles, an analogous baseline forms $B=A^2$ and computes $\frac{1}{2}\sum_{i<j}\binom{B_{ij}}{2}$, again requiring matrix multiplication and processing $\Theta(n^2)$ entries in MPC. Quadrangle enumeration is even costlier when materialized as pairs of common neighbors for every vertex pair, yielding up to $\Theta(n^4)$ candidates.


A line of prior work~\cite{araki2021secure,mazloom2018secure,nayak2015graphsc,graphiti,peng2024mapcomp,zou2025ringsg,zou2024cognn} explores secure graph computation through message passing, which is well suited to iterative aggregation tasks such as PageRank, where vertices exchange constant-sized updates. Triangle and quadrangle detection are \textit{fundamentally different}: they require identifying topological patterns, not simply aggregating per-node values. Even in the cleartext setting, the number of subgraphs that must be inspected grows rapidly with vertex degree, making these computations substantially more complex than message aggregation.
Consequently, these works are intrinsically ill-equipped for our task. The closest work, Oryx~\cite{zhong2024oryx}, supports private cycle detection but suffers from several drawbacks: (1) leaks significantly more information (all $k$-hop paths), (2) assumes a semi-honest threat model, and (3) incurs high memory overhead that limits practicality to very low-degree graphs. These limitations motivate a specialized approach that scales to realistic graphs under a stronger threat model.
\\\noindent\textbf{Our Contributions.} To address this gap, we introduce \sysname\footnote{\textbf{Tri}angle and \textbf{Q}uadrangle \textbf{E}numeration \textbf{T}ool}, a system for privacy-preserving triangle and quadrangle detection over distributed graphs. \sysname~operates in the 2+1 server model, a widely adopted MPC architecture due to its practical efficiency~\cite{wagh2022pika,boyle2025preprocessing,agarwal2024secure,cryptoeprint:2020/1392,cryptoeprint:2019/1095}. Data owners secret-share their local graphs between two non-colluding servers, which jointly execute the MPC protocol. A third auxiliary server, the dealer, participates \textit{only} in the data-independent offline phase to generate correlated randomness and plays no role online. \sysname~operates under the malicious threat model and supports both counting and enumeration of triangles and quadrangles.

The core challenge is that graph structure itself is sensitive and must remain hidden. An MPC protocol must therefore execute data-obliviously, performing the \textit{same} work regardless of the true topology. Yet real-world graphs are sparse, making such uniformity enormously wasteful. This creates a fundamental tension: hiding structure requires uniform work, while scalability requires exploiting sparsity. \sysname~addresses this through a suite of novel techniques:
\\\noindent \textbf{1. Leveraging Sparsity via Differentially Private Leakage.}
Revealing differentially private noisy degrees unlocks a major efficiency gain: \sysname~bounds per-vertex processing by $n\hat{d}_{max}$ edges (where $\hat{d}_{max}$ is the noisy maximum degree), rather than the $n^2$ worst case. Since noisy degrees are public, we can organize computation around the actual graph density without hiding this information inside MPC.\footnote{While \cite{mazloom2018secure,mazloom2020secure} also use DP leakage, their setting and goals differ substantially from ours; see Sec.~\ref{sec:relatedwork} for details.}
\\\noindent \textbf{2. Reduced No-Ops via Noisy-Degree Ordering.}
The above already saves substantial work, but real-world graphs have long-tailed degree distributions, where most vertices have degrees far below $\hat{d}_{max}$. To exploit this, we impose a noisy-degree-based vertex ordering. By processing vertices in this order and retiring edges as they are consumed, \sysname{} reduces the number of processed neighbors to $\sum_{i=1}^n \hat{d}_i$, proportional to the number of edges in the graph.
\\\noindent \textbf{3. Hiding Access Patterns with Constant Fake Access.}
Graphs can leak information through access patterns even when values are secret shared. A naive defense replaces each lookup with a full linear scan over all $n$ vertices, which is prohibitively slow; ORAM avoids full scans but incurs $O(\log n)$ overhead per access with large constants. We observe that our access pattern is highly structured, making full ORAM machinery unnecessary. Exploiting this structure, we design a technique that pairs each real access with exactly \textit{one} fake access, efficiently hiding access patterns.
\\\noindent\textbf{4. Efficient Malicious Security.} Beyond standard malicious behavior in MPC, graphs introduce a unique vulnerability: each undirected edge is reported \emph{twice}, once by each endpoint. When the endpoints belong to different data owners, a malicious party can falsify an edge, creating asymmetric neighbor lists and corrupting the computation. Naively checking all $n^2$ adjacency entries for symmetry is prohibitively expensive. Instead, we introduce a consistency check with only $O(|E|\log|E|)$ cost.
\\\noindent\textbf{5. End-to-end Optimized MPC.} We incorporate several implementation level optimizations to make \sysname~practical at scale. A major challenge is the large number of non-linear operations (equality tests, comparisons) required for cycle detection, which are notoriously expensive in secret-sharing–based MPC. We address this by using function secret sharing (FSS) to evaluate these operations, pushing nearly all heavy computation into the offline phase and leaving an extremely lightweight, one-round online execution--yielding substantial performance improvements~\cite{boyle2016function,cryptoeprint:2020/1392,cryptoeprint:2019/1095,wagh2022pika}. Finally, we exploit parallelism in our protocols: we show how to parallelize the most computation-intensive components, enabling highly scalable and efficient performance.
\\\textbf{Experimental Results.}
We compare \sysname{} against Oryx~\cite{zhong2024oryx}, an adjacency-matrix baseline, and three memory-access baselines (\textit{ORAM}, \textit{OSAM}, and \textit{Shuffle}). \sysname{} consistently achieves the best performance despite operating under a stronger malicious threat model than Oryx, \textit{ORAM}, and \textit{OSAM}. End-to-end execution takes $<$2.3 minutes for triangle enumeration and $<$2.1 minutes for quadrangle detection. Compared with graph-based baselines, \sysname{} improves runtime by up to $9.2\times$ (Oryx) and $4\times10^5\times$ (adjacency matrix). Compared with generic memory-access baselines, it achieves up to $4.7\times10^5\times$ lower runtime while reducing communication by up to $6.9\times10^3\times$.

\section{Background}\label{sec:background}

\noindent \textbf{Linear Secret Shares.} 
Linear secret shares (LSS) is an MPC technique that allows mutually-distrusting parties to securely compute over secret inputs. We use $[x]$ to denote a linear secret sharing of an input $x \in \mathbb{F}$. Each party $P_i$ holds a share $[x]_i$ such that $\sum_{i=1}^k [x]_i=x$, where $k$ denotes the total number of participating parties.
LSS supports local linear operations. 
\\\noindent \textbf{Authenticated Secret Shares.} 
Here, we adopt the SPDZ-style~\cite{damgaard2010multiparty,keller2016mascot} authenticated secret shares (ASS), which ensures the integrity of the shared secrets. The approach involves augmenting the secret shares with an additional sharing of an information-theoretic message authentication code (IT-MAC). Specifically, each party shares $[\gamma]$ for a secret MAC key $\gamma \in \mathbb{F}$, and for a shared value $[x]$, they also share the corresponding MAC $[\gamma \cdot x]$. We denote $\langle x \rangle = ([x], [\gamma \cdot x])$
to be the authenticated secret sharing for a secret $x$ and $\langle x \rangle_i = ([x], [\gamma \cdot x]) \in \mathbb{F}^2$ as the authenticated secret share held by party $P_i$. ASS also supports local linear operations.
\\\noindent\textbf{Function Secret Share (FSS).}
A two-party FSS~\cite{boyle2015function,boyle2016function} scheme splits a function $f$ into succinct function shares (keys) such that any single share reveals no information about $f$, yet combining the parties’ evaluations at an input $x$ reconstructs the correct output $f(x)$. We focus on distributed point functions (DPFs), which implement FSS for point functions $f_{\alpha,\beta}$ defined by $f_{\alpha,\beta}(\alpha)=\beta$ and $f_{\alpha,\beta}(\alpha')=0$ for all $\alpha' \neq \alpha$. Conceptually, such a function corresponds to a one hot vector $\mathbf{e}=(0,\ldots,0,\beta,0,\ldots,0)$ with $\beta$ at position $\alpha$. Given their keys, party $b$ locally evaluates to obtain $y_b(x)$ such that, for all $x$, $y_0(x)+y_1(x)=f_{\alpha,\beta}(x)=\mathbf{e}[x]$. (See App.~\ref{app:FSS}.)

\noindent\textbf{Edge Differential Privacy.}
We focus on the \textit{local} model of differential privacy (DP), where each data-owning party perturbs its own data before sending it to an untrusted aggregator. For graphs, the standard local guarantee is edge-local DP~\cite{Graph7,Graph11,qin2017generating,imola2022robustness,kolluri2022lpgnet,eden2025triangle,279970,hillebrand2024cycle}, which protects the existence of any individual edge. Thus, an adversary observing the output cannot distinguish between graphs differing in a single edge.
\begin{definition}\label{def:eldp}
A randomized mechanism $\mathcal{M}$ is $(\varepsilon,\delta)$-\textit{edge differentially private} if, for any two graphs $G$ and $G'$ differing in exactly one edge and any measurable output set $S$,
\begin{gather}\label{eq:ledp}
\Pr[\mathcal{M}(G)\in S]
\leq e^{\varepsilon}\Pr[\mathcal{M}(G')\in S]+\delta.
\end{gather}
\end{definition}
In the local model, $G$ is the subgraph held by one party, which may be as small as a single adjacency list. While the classical Laplace mechanism is the canonical tool for achieving DP, we require \textit{non-negative}, \textit{bounded} noise. We therefore use the \emph{truncated shifted discrete Laplace} distribution, $\text{TSDLap}(\lambda,t)$, supported on ${0,\ldots,2t}$ with probability proportional to $\exp(-|x-t|/\lambda)$. Specifically, if each party reports
$\tilde d_v=d_v+\eta$ for vertex $v$ in its subgraph, where
$
\eta\sim\text{TSDLap}\!\left(\frac{2}{\varepsilon},
t=\Big\lceil2+\frac{2}{\varepsilon}\log\!\left(\frac{2}{\delta}\right)\Big\rceil\right),
$
the mechanism satisfies $(\varepsilon,\delta)$-edge local DP~\cite{Bell22,cryptoeprint:2024/469}.


\section{Problem Statement}\label{sec:setting}
\noindent \textbf{Notations.} Let $\G = (V, E)$ be an undirected, unweighted graph, where $V$ and $E$ denote the sets of vertices and edges, respectively. We assume $|V| = n \in \mathbb{N}$, so $V = \{v_1, \ldots, v_n\}$. For each vertex $v \in V$, let $\N_v$ denote its list of neighbors. Additionally, let $d_v$ and $\hat{d}_v$ be its true and  noisy degree, respectively. Both $n$ and the noisy degrees $\hat d_v$ (once published, Sec.~\ref{sec:datacollection}) are \emph{public}, known to every data-owning party and both servers; only the true edges are kept secret.
\\\textbf{Setting.} We consider a distributed setting where $\G$ is \textit{partitioned} across $m$ parties ($2 \le m \le n$). Each party $\Pa_i$ holds a disjoint subset of vertices $V_i \subseteq V$ and their incident edges, with
$
V = V_1 \cup \cdots \cup V_m,
\quad V_i \cap V_j = \emptyset \text{ for } i \neq j.
$
Thus, each party has only a \textit{local} view of the graph, given by
$
G_i=(V_i,E_i), \quad
E_i={(u,v)\in E\mid u\in V_i},
$
equivalently, all edges appearing in neighbor lists $\ell_v$ for $v\in V_i$. No party knows the complete graph. Let $A_G$ denote the adjacency matrix of $\G$, where $A_{ij}=1$ iff an edge exists between $v_i$ and $v_j$, and $0$ otherwise.

As discussed in Sec.~\ref{sec:intro}, this model captures diverse deployments. At one extreme, $m=n$ captures distributed social networks such as Mastodon, Diaspora, and PeerTube, where each party represents a node and edges encode ``follower'' relationships. At the other, small $m$ captures organizations collaborating over partitioned networks, such as two Internet service providers holding different portions of a network topology, with nodes representing routers or network endpoints.
\\Our goal is to compute a statistic over the global graph $\G$. Specifically, we aim to: (1) \textit{count} the total number of triangles (3-cycle) and quadrangles (4-cycle) in $\G$; and (2) \textit{enumerate} i.e., list the vertices in each triangle and quadrangle. We discuss how to extend \sysname{} to additional tasks in App.~\ref{app:extension}.

\sysname~operates in the (2+1)-party setting, a well-established MPC model known for its practical efficiency~\cite{wagh2022pika,boyle2025preprocessing,agarwal2024secure,cryptoeprint:2019/1095,cryptoeprint:2020/1392}. The protocol has two phases: an input-independent offline phase that can be executed in advance, and an online phase that performs cycle detection once the graph data is available. During the online phase, the $m$ data-owning parties secret-share their local graph data between two \emph{non-colluding, untrusted} servers, $\Ser_0$ and $\Ser_1$, which jointly execute the MPC protocol for cycle detection. An auxiliary server, or ``dealer,'' participates \textit{only} offline to generate correlated randomness and is not involved in online execution.

The dealer generates two types of correlated randomness. First, it enables authenticated secret sharing by sampling a global MAC key $\alpha$ and a random value $r$, then distributing additive shares $\alpha_0+\alpha_1=\alpha$ and $r_0+r_1=r$ to the servers. It additionally provides $z_b=\alpha_b r+\delta_b$, where $\delta_0+\delta_1=0$. To secret-share a data-owning party's input $x$, server $\Ser_b$ sends $r_b$ to the party, which returns $y=x-(r_0+r_1)$ (optionally encrypted to hide it from the dealer). Server $\Ser_0$ sets $[x]_0=r_0$ and $[\alpha x]_0=\alpha_0y+z_0$, while $\Ser_1$ sets $[x]_1=y+r_1$ and $[\alpha x]_1=\alpha_1y+z_1$. Second, the dealer provides correlated randomness for our optimized protocol implementation (Sec.~\ref{sec:optimizations}).
\\\noindent\textbf{Threat Model.} \sysname~operates in the standard malicious security model with abort~\cite{10.1145/28395.28420,cryptoeprint:2000/067}.
In the \textit{offline} phase, either (1) the dealer may be malicious and provide ill-formed correlated randomness that must be verified, or (2) one of the two servers may be malicious. In the \textit{online} phase, at most one server may be corrupted. Any number of data-owning parties may also be malicious and collude with the corrupt server.
Unlike tabular data, each graph edge is reported by both endpoints. A malicious data owner may therefore create inconsistent neighbor lists by misreporting an edge, breaking correctness. Our protocols detect such inconsistencies and operate only on \emph{well-formed} graphs. If both endpoints are malicious and colluding, however, they may consistently misreport an edge. The resulting graph remains well-formed, and \sysname~computes correctly on this reported graph. Without a party holding the ground-truth global graph, such coordinated falsification is inherently undetectable and therefore outside our threat model. Without a third party holding the ground-truth global graph (precisely the distributed setting we consider), detecting such coordinated falsification is beyond the purview of cryptography. This limitation is inherent and therefore outside our threat model.

\section{Can we use Generic MPC?} 
\label{sec:MPC}
Before presenting our protocol, we highlight why generic MPC is ill-suited for our task. MPC reveals nothing beyond the output (here, only the triangles or quadrangles). This conflicts with standard graph algorithms: because \textit{structural information} must remain hidden, the protocol cannot follow the actual topology. Instead, it must perform \textit{data-oblivious} operations, including ``no-ops,'' on \textit{every} node and edge. For example, finding a vertex's neighbors requires testing against all vertices to avoid revealing its true degree.

Now consider triangle detection. A direct approach enumerates all triples $(i,j,k)\in[n]^3$, $i\neq j\neq k$, and checks whether $(v_i,v_j),(v_i,v_k),(v_j,v_k)$ exist via $A_{ij}A_{ik}A_{jk}$ (equivalently, $|C_3|=\frac{1}{6}\mathrm{Tr}(A^3)$). This is straightforward under MPC since matrix multiplication is data-oblivious, but requires $\Theta(n^3)$ secure multiplications. Although sub-cubic algorithms exist~\cite{alman2024asymmetryyieldsfastermatrix}, they are ``galactic algorithms"~\cite{galactic}, whose large constants and memory footprint make them impractical~\cite{dumas2019secure}.

To tackle secure graph computation, prior work~\cite{araki2021secure, mazloom2018secure, nayak2015graphsc, graphiti} has explored topology-hiding message passing, where nodes exchange messages along edges and update local state. This works well for iterative aggregation tasks such as PageRank, where messages are constant-size summaries. Cycle detection, however, is a \emph{pattern-matching} task that requires reasoning about graph structure. For triangle counting, nodes must effectively exchange entire neighbor lists and compute their intersections, requiring $O(n^2)$ operations per node and $O(n^3)$ overall\footnote{Oryx~\cite{zhong2024oryx} modifies this slightly but leaks all $k$-hop paths and assumes a bounded maximum degree, yet still falls short of \sysname{} (see Sec.~\ref{sec:evaluation}).}. The same limitation extends to quadrangle detection. Alternatively, triangle counting can be expressed as a SQL join-counting query, but even in the best case incurs $O(n^2d^3\log(n^2d^3))$ time on a $d$-regular graph due to exhaustive join padding and oblivious resizing. As shown in App.~\ref{app:sql}, this is asymptotically worse than \sysname{}, motivating our specialized protocols.



\section{\sysname: Triangle Detection}\label{sec:name}
As noted earlier, naive approaches scan $\Theta(n^3)$ triples for triangles and $\Theta(n^4)$ quadruples for quadrangles. Yet real graphs are often \emph{sparse}, making much of this work unnecessary. \sysname{} exploits this sparsity while preserving strong security guarantees. We first present triangle detection; quadrangle detection follows the same principles.
\begin{mybox2}{\algotitle{alg:tri_plain}{\textsc{TriangleDetect}}}
\footnotesize
\textbf{Input:}  Graph $\G=(V,E)$ with neighbor lists $\{\N_v\}_{v\in V}$ \\
\textbf{Output:} Triangle count $\mathit{cnt}$ and triangle list $\mathcal{L}$.
\begin{enumerate}[nosep, label=\arabic*.]
\item $\mathit{cnt} \gets 0$, \;\; $\mathcal{L} \gets [\,]$
\item Sort vertices as $v_1,\ldots,v_n$ such that $d(v_1)\ge \cdots \ge d(v_n)$ (break ties by id), and relabel them so that $v_i$ has id $i$.
\item \textbf{For} $i\in[n]$:
\item \hspace*{1.2em} \textbf{For each} $j \in \N_i$ with $j > i$:
 \item \hspace*{2.4em} $\mathit{cnt} \gets \mathit{cnt} + |l_{j}\cap l_i|$
\item \hspace*{2.4em} $\mathcal{L}.\text{Append}(i,j,k)$ where $k \in \{l_{j}\cap l_i\}$
\item $T \gets \mathit{cnt}/3$
\item \Return $(T,\mathcal{L})$
\end{enumerate}
\end{mybox2}
\subsection{Technical Overview}\label{sec:overview:naive}
We begin with a simplified algorithm (\autoref{alg:tri_plain}) that forms the backbone of our secure protocol. This is unsuitable for MPC as written because it leaks graph topology, but it provides the intuition for our data-oblivious protocol.
Rather than examining every vertex triple, we exploit sparsity by operating only on existing edges. A triangle $(v,u,w)$ exists exactly when $w\in\N_v\cap\N_u$. Thus, for each edge $(v,u)\in E$, we count (or enumerate) triangles by computing $\N_v\cap\N_u$.
\begin{figure*}[]
\centering
\includegraphics[width=\textwidth]{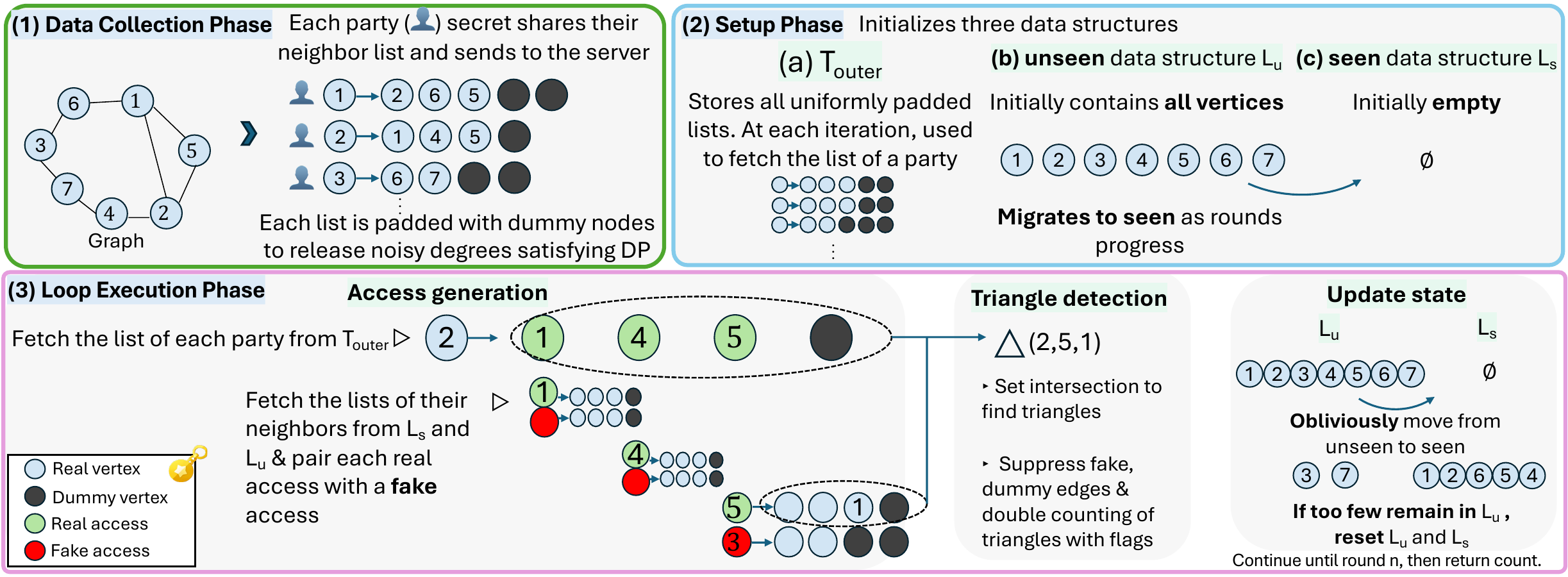}
\caption{\textbf{High-level overview of \sysname.} \sysname{} runs in three phases. In (1) \textit{data collection}, each data owner locally perturbs vertex degrees to obtain DP noisy degrees ($\hat d_v$) and submit secret-shared neighbor lists padded to the noisy degrees (\autoref{fig:degree}). In the one-time (2) \textit{setup phase}, the servers first check the well-formedness of the submitted neighbor lists to detect malicious behavior (omitted for simplicity; see \autoref{alg:wfcheck}). At a high level, \sysname{} follows the two-loop structure of \autoref{alg:tri_plain}. To enable this, the servers initialize three data structures: (i) $\To$, used to fetch neighbor lists in the outer loop; and (ii) $\Seen$ and $\US$, used in the inner loop to track vertices already accessed and not yet accessed, respectively. As the outer loop progresses, vertices move from $\US$ to $\Seen$; thus, $\Seen$ grows while $\US$ shrinks. In (3) \textit{loop execution}, the outer loop fetches the neighbor list of vertex $i$ from $\To$, while the inner loop retrieves the corresponding neighbors. The key insight for hiding the resulting access pattern is to \textit{equalize} the servers' views of accesses to $\Seen$ and $\US$ vertices. To this end, \sysname{} generates $\hat d_i$ fake accesses through a custom MPC protocol such that they are indistinguishable from the $\hat d_i$ real accesses (Fig.~\ref{fig:seen-access}), where $\hat d_i$ is the noisy degree of vertex $i$. The protocol then computes the triangle count by securely intersecting the retrieved neighbor lists in the two loops; appropriate secret-shared flags suppress contributions from dummy entries and fake accesses, ensuring correctness. Finally, \sysname{} updates the data structures for the next iteration, obliviously moving newly accessed vertices from $\US$ to $\Seen$ (Fig.~\ref{fig:update-state}). When the $\US$ pool becomes too small to generate the required fake accesses, the data structures are securely reset and reinitialized. After $n=|V|$ rounds, \sysname{} returns the triangle count.}
\label{fig:protocol-trace}
\end{figure*}

An important ingredient in our approach is degree-based ordering under MPC. We sort vertices by decreasing degree and re-index them so that $v_i$ has id $i$ (\autoref{fig:degree}); we use these re-indexed IDs throughout the rest of the paper. Under this ordering, \autoref{alg:tri_plain} considers only neighbors $j\in\N_i$ with $j>i$ and, for each such edge, computes $\N_j\cap\N_i$, ensuring every edge is processed once. As discussed in Sec.~\ref{sec:MPC}, data-oblivious MPC incurs many no-ops to hide graph topology. Processing each edge in one direction substantially reduces this overhead, as detailed in Sec.~\ref{sec:protocol}. Each triangle is encountered three times, once per edge, so we divide the count by three. Duplicates in triangle list $\mathcal{L}$ are removed after publication.
\noindent\textbf{Key Ideas.} Despite its efficiency, the algorithm above leaks graph topology and cannot be used directly under MPC. \sysname~introduces four ideas that preserve its structure while making the computation data-oblivious.
\\\textbf{1. Leaking Differentially Private Degrees.}
The usage of exact degrees in \autoref{alg:tri_plain} exposes sensitive information. Instead, we relax the security guarantee by leaking a \textit{differentially private estimate} of degrees. This leakage is formally bounded: DP ensures the adversary learns nothing about the presence of any individual edge (Def.~\ref{def:eldp}) in a neighbor list, yet it enables large performance gains. Prior MPC work has explored related leakage tradeoffs~\cite{10.1145/3133956.3134030,10115007,Groce2019CheaperPS,mazloom2018secure},\footnote{While \cite{mazloom2018secure,mazloom2020secure} use DP leakage in secure graph message passing, their use of DP differs fundamentally from ours. CARGO~\cite{liu2024cargo} reveals noisy degrees only to bound sensitivity for releasing DP triangle counts and still incurs $O(n^3)$ computation (see Sec.~\ref{sec:relatedwork} and App.~\ref{app:related-work}).} but the appropriate leakage is application-dependent. Here, noisy degrees provide three benefits: (1) \emph{sparsity}, by scaling work with the sum of noisy degrees rather than all $n^2$ potential edges; (2) \emph{communication}, by padding neighbor lists only to their noisy degrees rather than $n$; and (3) \emph{accuracy}, by tightening global sensitivity for cycle counting (Sec.~\ref{sec:DP}), reducing output noise. For enumeration, noisy degrees do not affect correctness: only correct triangles are identified. Notably, for large $m$, \sysname{} can provide \emph{stronger} privacy than secure message-passing frameworks~\cite{nayak2015graphsc,graphiti,kapoor2025mathsf}, which reveal $|V_i|+|E_i|$ for each data owner through input-list sizes. When $m=n$ and each party contributes one vertex, this reveals every vertex's exact degree; \sysname{} instead reveals only DP degree estimates.
\\\textbf{2. Hiding Access Patterns with Constant Fake Accesses.}
In \autoref{alg:tri_plain}, each inner-loop iteration retrieves the neighbor lists of the current vertex's neighbors. Secret sharing hides the accessed values but not the memory-access pattern, which can reveal graph structure. For example, repeated accesses across iterations can directly reveal common neighbors, with leakage accumulating over time. A naive defense is to scan all $n$ vertices for every lookup, hiding the access pattern but at prohibitive cost. The standard alternative is oblivious RAM (ORAM), specifically distributed ORAM in our setting,\footnote{Since both the access index and memory array are secret-shared.} but ORAM incurs $\Omega(\log n)$ overhead per access~\cite{goldreich1996software}. Despite theoretical advances, practical ORAM constructions remain costly due to large constants and can even underperform linear-scan approaches~\cite{cryptoeprint:2017/827}, as we also observe empirically (Sec.~\ref{sec:evaluation}).

Our key insight is that we do not need the full generality of ORAM. \autoref{alg:tri_plain} has a structured access pattern: in outer-loop iteration $i$, it performs exactly $\hat{d}_i$ memory accesses. The main leakage is how many target previously accessed (`seen') versus new (`unseen') vertices. We hide this by pairing each real access with a \textit{single} fake access. Thus, in iteration $i$, the servers observe exactly $\hat{d}_i$ accesses drawn uniformly from `seen' positions and $\hat{d}_i$ from `unseen' positions, independent of the true graph. This hides the access pattern with only constant-factor ($2\times$) overhead, rather than linear scans or ORAM. We operationalize this idea with MPC-friendly data structures that make fake and real accesses indistinguishable.
\\\textbf{3. Reduced No-ops via Degree Ordering.}
When retrieving neighbor lists, we must hide \emph{which} neighbors are accessed, since list lengths can reveal vertex identities. A straightforward defense: padding every list to the global maximum noisy degree $\hat{d}_1$ is secure but wasteful, especially since most real graphs are skewed, with many low-degree vertices and only a few near the maximum. Our key insight is that each undirected edge need only be processed once, when it first appears, and can be ignored thereafter. Since the outer loop visits vertices in decreasing noisy-degree order, neighbor lists fetched in later iterations only need padding up to the maximum \emph{remaining} degree among unprocessed vertices, not $\hat{d}_1$. The challenge is retiring edges in a data-oblivious way. We address this by using the noisy degree ordering to process each edge only in a single direction determined by its endpoints' degrees.
\\\textbf{4. Efficient Malicious Security.}
Graphs introduce a unique challenge beyond standard malicious behavior in MPC: each undirected edge is reported by \textit{both} endpoints. A malicious data owner can therefore create inconsistencies across neighbor lists by misreporting an endpoint. Hence, we require the reported graph to be \emph{well-formed}. If both endpoints are malicious and consistently falsify an edge, the graph still remains well-formed, and our protocol computes correctly on that reported graph. Detecting such collusion is inherently outside cryptographic guarantees (Sec.~\ref{sec:setting}). Naively checking well-formedness requires verifying $A_{jk}=A_{kj}$ across all $n^2$ possible entries. Instead, we exploit a key invariant: every true edge should appear exactly twice across all neighbor lists, and we show how to enforce this in only $O(|E|\log|E|)$ work.
\vspace{-1mm}

\subsection{Protocols}
\label{sec:protocol}
\noindent\textbf{Workflow.} The full workflow is illustrated in \autoref{fig:protocol-trace}. \sysname{} runs in three phases. In (1) \textit{data collection}, where data owners locally perturb vertex degrees and publish noisy estimates. Each party then submits secret-shared neighbor lists for its vertices, padded to the noisy degrees. After this, data owners can go offline, and the remainder runs solely between the two servers. Next, in the one-time (2) \textit{setup phase}, the servers first check the well-formedness of the neighbor lists to detect malicious behavior. Then, they store the padded lists in $\To$ and initialize the global state: a collection of MPC-friendly data structures used throughout the computation.

Next, \sysname~enters the (3) \textit{loop-execution} phase, which follows the two-loop structure of \autoref{alg:tri_plain}: the \emph{outer loop} iterates over vertices in the noisy degree order ($i$ denotes the $i$-th vertex). The \emph{inner loop} retrieves the neighbors for each vertex $i$, where $\hat{d}_i$ is its noisy degree. To hide the true access patterns, every real neighbor access is paired with a \textit{fake} access (\autoref{fig:seen-access}). Consequently, the servers observe exactly $\hat{d}_i$ random accesses drawn from `seen' and $\hat{d}_i$ random accesses from `unseen,' ensuring that the observed access pattern is oblivious of the neighborhood structure. For each accessed vertex, \sysname{} computes common neighbors with vertex $i$, and incorporates flags that suppress contributions from fake accesses. Next, \sysname~\textit{updates} the global state. 

\begin{mybox2}{\algotitle{alg:protocol}{\sysname: Protocol for Triangle Enumeration}}
\footnotesize
\setlength{\parindent}{0pt}

\textbf{Parameters:} $\epsilon, \delta$; Privacy parameters

\vspace{1mm}
\textit{Data Collection Phase} 

\textbf{Each} data owning party $\Pa_j$, $j \in [m]$

\begin{enumerate}[nosep, label=\arabic*.]
\setcounter{enumi}{0}

\item \textbf{For} each vertex $v$ in $V_j$
\item \hspace{0.25cm}
      Sample $\eta_v \sim
      \Op{TSDLap}
      \left(
      2/\epsilon,\,
      t=\left\lceil
      2+\frac{2}{\epsilon}\log\frac{2}{\delta}
      \right\rceil
      \right)$
\item \hspace{0.25cm} Publish $\hat d_v = d_v + \eta_v$
\item \textbf{EndFor}

\item[] \textcolor{blue}{$\rhd$ Re-index vertices according to noisy degrees}

\item Sort all $n$ vertices as $v_1,\ldots,v_n$ such that
      $\hat d_{v_1} \geq \hat d_{v_2} \geq \cdots \geq \hat d_{v_n}$
      (ties broken by ID)
\item Re-index $v_i$ as $i$

\item \textbf{For} each vertex $i$ in $V_j$ (after re-indexing)
\item \hspace{0.25cm}
      Sample $\eta_i$ dummies
      $D_i \subset_R [n+1,n+2t]$
\item \hspace{0.25cm} $l_i.\Op{Append}(D_i)$
\item \hspace{0.25cm} Sort all vertices in $l_i$ according to $\prec$
\item \hspace{0.25cm} $\Op{SecretShare}(l_i)$
\item \textbf{EndFor}

\end{enumerate}

\vspace{1mm}
\textbf{Each} server $\Ser_b$

\vspace{1mm}
\textit{Setup Phase}

\begin{enumerate}[nosep, label=\arabic*.]
\setcounter{enumi}{12}

\item Set $\langle \To[i] \rangle = \langle l_i \rangle$

\item \textbf{If}
      $\Proc{WFCheck}(\langle \To \rangle, \{\hat d_i\}_{i\in[n]})=0$,
      then \textbf{Abort}

\item
$\begin{aligned}[t]
(\langle \Seen \rangle,\langle \US \rangle)
\leftarrow{}&
\Proc{InitializeState}(
    \langle \To \rangle,
    \hat d_1, n, \Op{isReset}=0)
\end{aligned}$

\end{enumerate}

\vspace{1mm}
\textit{Loop Execution Phase}

\begin{enumerate}[nosep, label=\arabic*.]
\setcounter{enumi}{17}

\item $\Op{Flag} \gets 1$

\item \textbf{For} $i \in [n]$
      \hfill
      \textcolor{blue}{$\rhd$ Outer loop in order $\prec$}

\item \hspace{0.25cm}
      $\langle l_i \rangle \gets \langle \To[i] \rangle$

\item \hspace{0.25cm}
      \textbf{If} $\neg\Op{Flag}$
      \hfill
      \textcolor{blue}{$\rhd$ Generate both real and fake accesses}

\item \hspace{0.5cm}
      $\langle R_s \rangle
      \gets
      \Proc{GenSeenAccesses}
      (\langle l_i\rangle,\langle \Seen\rangle,\hat d_i)$

\item \hspace{0.5cm}
      $\langle R_u \rangle
      \gets
      \Proc{GenUnseenAccesses}
      (\langle l_i\rangle,\langle \US\rangle,\hat d_i)$

\item \hspace{0.25cm} \textbf{Else} \hfill
      \textcolor{blue}{$\rhd$ Round 1: only true accesses, all tagged real}

\item \hspace{0.5cm}
      $\langle R_s\rangle\gets\varnothing,\quad
       \langle R_u\rangle\gets\{(1,v): v\in \langle l_i\rangle\}$

\item \hspace{0.25cm}
      \textbf{For} $j\in[\hat d_i]$
      \hfill
      \textcolor{blue}{$\rhd$ For all real and fake neighbors}

\item \hspace{0.5cm}
$\begin{aligned}[t]
\langle \Li\rangle.\Op{Append}\big(
&\Proc{EnumerateTri}(
i,\langle l_i\rangle,\langle R_s[j]\rangle,\langle\Seen\rangle)
\big)
\end{aligned}$

\item \hspace{0.5cm}
$\begin{aligned}[t]
\langle \Li\rangle.\Op{Append}\big(
&\Proc{EnumerateTri}(
i,\langle l_i\rangle,\langle R_u[j]\rangle,\langle\US\rangle)
\big)
\end{aligned}$

\item \hspace{0.25cm} \textbf{EndFor}

\item \hspace{0.25cm}
      $\Op{Flag}\gets0$
      \hfill
      \textcolor{blue}{$\rhd$ Use fake accesses until reset}

\item \hspace{0.25cm}
      $\Proc{UpdateState}(\Seen,\US,R_u,i)$

\item \hspace{0.25cm}
      \textbf{If} $\cnt_u<\hat d_{i+1}$
      \hfill
      \textcolor{blue}{$\rhd$ Reset state}

\item \hspace{0.5cm}
$\begin{aligned}[t]
(\langle\Seen\rangle,\langle\US\rangle)
\leftarrow{}&
\Proc{InitializeState}(
\langle\To\rangle,\hat d_i,n,\\
&\Op{isReset}=1)
\end{aligned}$

\item \hspace{0.5cm} $\Op{Flag}\gets1$

\item \textbf{EndFor}

\item \textbf{Return}
      $\rec\!\left(\Fshuffle(\langle\Li\rangle)\right)$

\end{enumerate}

\end{mybox2}

A technical subtlety is that generating fake accesses consumes entries from the `unseen' set. Consumed unseen vertices are promoted to `seen'. Over time, the unseen pool may shrink until it can no longer support the required number of fake accesses. When this happens, the servers \textit{reset} the seen/unseen data structures and continue. The details of each phase are described in the following subsections.

\begin{figure}[t]
    \centering
    \includegraphics[width=1\linewidth]{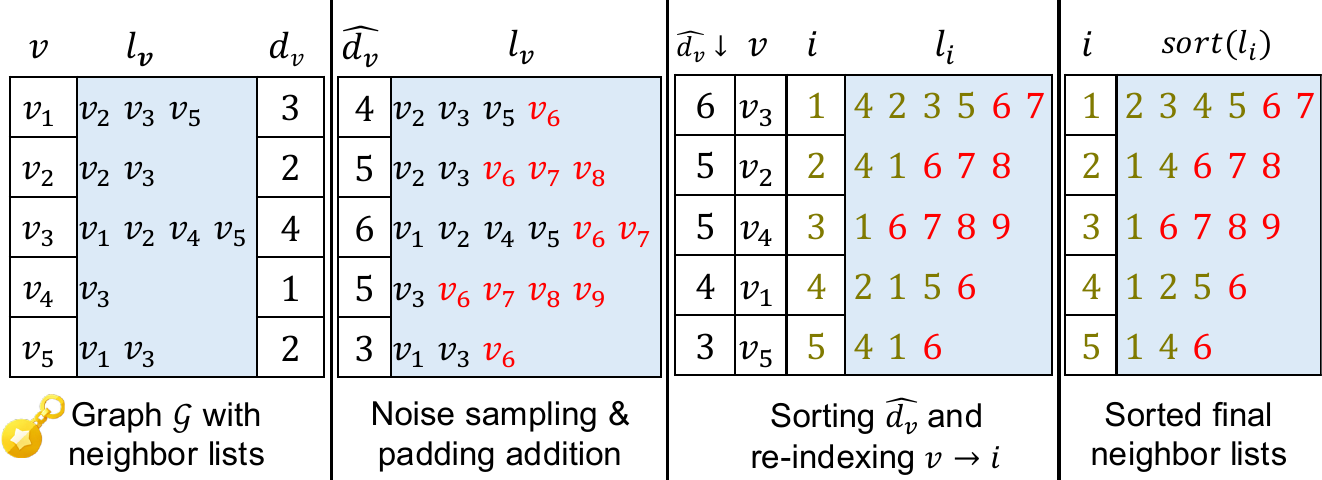}
    \caption{\sysname{} \textbf{Data Collection}: Re-indexing vertices by noisy degree.
Given $\G$ with neighbor lists, each party adds noise to local degrees and publishes $\hat d_v$.
Each party then pads its neighbor lists with dummy IDs (values $>5$ here), and finally re-indexes vertices by the degree-sorted order induced by $\hat d_v$.}
\vspace{-5mm}
    \label{fig:degree}
\end{figure}
\subsubsection{Data Collection}\label{sec:datacollection}
Recall that $m$ data-owning parties $\Pa_j$ each hold a disjoint subset of vertices $V_j \subseteq V$ and their incident edges. Each party locally perturbs the degree of every owned vertex $v \in V_j$ with non-negative, integral, bounded noise. Specifically, for each $v$, the party samples
$\eta_v \sim \mathsf{TSDLap}\!\left(\frac{2}{\epsilon},
t=\left\lceil 2 + \frac{2}{\epsilon}\log\frac{2}{\delta}\right\rceil\right)$
and publishes $\hat{d}_v=d_v+\eta_v$, making the noisy degree public. These degrees induce a global ordering: let $v_1,\ldots,v_n$ be sorted by decreasing noisy degree, breaking ties by original ID. This defines the public total order $\prec$, where
$v_i \prec v_j \Longleftrightarrow \hat d_{v_i}>\hat d_{v_j}$ (or they tie and $v_i$ has smaller ID).
Vertices are then \emph{re-indexed} so that $v_i$ receives public logical identifier $i$, which we use throughout (Fig.~\ref{fig:degree}). As we show later, this public re-indexing allows \sysname{} to eliminate many no-op operations, substantially reducing the cost of subsequent computations.

Next, each party submits secret-shared neighbor lists for its vertices, padded to their published noisy degrees. We reserve $[n+1,n+2t]$ for dummy vertices, where $2t$ is the maximum noise. For each $v$, the party samples $\eta_v$ dummy IDs, appends them to $l_v$, remaps real neighbors to their re-indexed IDs, and sorts the list. Thus, real neighbors appear first (sorted by $\prec$), followed by dummy entries (Fig.~\ref{fig:degree}).
\begin{tcolorbox}[
colback=gray!10!white, breakable,
colframe=black,
boxrule=0.6pt,
arc=1mm
]
\footnotesize
The key contribution of the data collection phase is the observation that each party can locally perturb and \textit{publicly} release DP-noisy degrees. This enables the high-level two-loop structure of Alg.~\ref{alg:tri_plain}: the outer-loop progression becomes public, unlocking \sysname's ability to leverage the graph's sparsity. Only topology-dependent inner-loop accesses need to protected throughout the rest of the computation. 
\end{tcolorbox}

\subsubsection{Setup}
\label{sec:protocol:setup}
The servers run a one-time setup phase to (1) ensure all submitted neighbor lists are well-formed and (2) initialize the global state for loop execution.

\noindent\textbf{Well-formedness Check.}
A malicious client may submit malformed neighbor lists to corrupt the computation. The servers therefore perform a one-time MPC-friendly check with three components:
\\
(i) \textit{Sortedness.} As in Sec.~\ref{sec:datacollection}, each list must be strictly increasing. The servers securely scan each list and verify that every entry lies in $[n+2t]$ and exceeds its predecessor.
\\
(ii) \textit{Bounded Dummy Padding.} Since lists are padded to their published noisy degrees, a malicious party could inflate degrees to force arbitrarily long lists. Although dummies are later filtered out, they increase cost. We therefore use \Op{TSDLap}, which bounds noise to $[0,2t]$, and enforce at most $2t$ dummy identifiers (entries $\ge n+1$) per list.
\\
(iii) \textit{Global Edge Consistency.}
Every real edge $(u,v)$ in an undirected graph must appear twice, once in each endpoint's list. A corrupted party can violate this invariant across parties, breaking correctness while passing per-list checks. Rather than comparing all $O(n^2)$ matrix entries, \sysname{} uses a counting-based test. Each edge receives a canonical encoding $\mathcal{E}$ identical for both orientations, e.g.,
$\mathcal{E}(u,v)
  = \min(u,v)\cdot B + \max(u,v); B = 10^{\lceil \log_{10} n \rceil}$.
The servers build a secret-shared list $\Q$ of these encodings and verify that every real-edge encoding appears exactly twice. Dummy neighbors ($>n$) are obliviously mapped to a placeholder $p$. Let $e=\sum_i \hat d_i$; if padding makes $e$ odd, one placeholder is appended so $|\Q|$ is even. The servers then securely shuffle, sort, and scan $\Q$: every consecutive pair $\Q[2k]$ and $\Q[2k+1]$ must match. Any mismatch indicates an inconsistent edge and causes an abort.
\noindent\textbf{Initialization of Global State.}
\sysname{} maintains three MPC-friendly data structures forming the global state:
\\
(i) \textit{Outer-Loop Structure} $\To$.
The outer loop iterates over vertices in the public order $\prec$ induced by the noisy degrees. The servers initialize
$\langle \To[i] \rangle = (i,\; \langle l_i \rangle)$,
where $l_i$ is the secret-shared, padded neighbor list of vertex~$i$.
\\
\noindent(ii) \textit{Tracking Unseen Vertices} $\US$. During iteration $i$, the protocol must fetch every vertex in $l_i$ without revealing which entries are real neighbors. Even revealing the noisy degree of a fetched vertex could reveal its identity, since noisy degrees are public. To prevent this, the servers pad every neighbor list to the global maximum noisy degree $\hat d_1$ and append $2t$ dummy neighbor lists for identifiers in $[n{+}1,n{+}2t]$. Thus, fetched records have identical lengths, and dummy lookups are indistinguishable from real ones.

To support $O(1)$ secure lookups over this padded collection, the servers build a secret-shared hash table $\Tu$ (Alg.~\ref{alg:hashtable}). It uses a secret-shared universal hash function $H(y)=(ay+b)\bmod p$, where random coefficients $a,b$ are additively shared between the servers. Thus, to either server, hash locations are information-theoretically uniform, yielding an oblivious hash table~\cite{10.1145/100216.100289,cryptoeprint:2024/011}.

The servers also maintain a secret-shared list $D$ of all vertex IDs and a counter $\cnt_u$ tracking the number of unseen vertices. The key invariant is that $D[1..\cnt_u]$ contains exactly the currently unseen vertices; initially, $\cnt_u=n$, so all vertices are unseen. As vertices are accessed, they are removed from this active prefix and promoted to the seen set. To support these removals efficiently, each entry in $\Tu$ stores a pointer to the corresponding vertex's position in $D$, enabling $O(1)$ deletion (see Sec.~\ref{sec:protocol:round} for details). The pointers are initialized alongside $\Tu$ by constructing $D$ in the same order, introducing no additional leakage. Finally, to sample unseen vertices without revealing their positions in $D$, the servers maintain a second secret-shared array $D_u$, containing a freshly shuffled copy of the active prefix $D[1..\cnt_u]$. Thus, $D_u$ is used for oblivious sampling, while $D$ supports efficient updates. Together, these form the unseen data structure $\US=(\Tu,D,D_u,\cnt_u)$.
\\
\noindent(iii) \textit{Tracking Seen Vertices} $\Seen$.
\sysname{} also maintains
$\Seen=(\Ts,D_s,cnt_s)$ for vertices retrieved in prior iterations, where $D_s$ stores seen vertex IDs, $\Ts$ stores their neighbor lists for $O(1)$ secure lookup, and $cnt_s$ tracks the number of seen vertices. Initially, $cnt_s=0$, $D_s$ and $\Ts$ are empty. As the protocol progresses, vertices migrate obliviously from $\US$ to $\Seen$ and are retrieved from $\Ts$ during loop execution.

\begin{tcolorbox}[
colback=gray!10!white, breakable,
colframe=black,
boxrule=0.6pt,
arc=1mm
]
\footnotesize
The main contributions of the setup phase are:
\begin{itemize}[leftmargin=1em,itemsep=0pt]
\item First, \sysname{} provides an efficient $O(|E|\log |E|)$ check for detecting malformed inputs, improving substantially over the naive $O(|E|^2)$ approach.
\item Second, \sysname{} introduces three custom data structures that enable efficient execution of the two-loop algorithm. Public outer-loop retrievals use a hash table $\To$. The key challenge--and contribution--is making inner-loop retrievals oblivious without generic ORAM. \sysname{} exploits a crucial structural insight: vertices need only be distinguished as previously accessed (``seen") or not yet accessed (``unseen"). This yields two specialized oblivious data structures: a growing seen pool and a shrinking unseen pool, initially containing all vertices.
\end{itemize}

\end{tcolorbox}

\subsubsection{Loop Execution}
\label{sec:protocol:round}
\sysname{} proceeds in $n$ rounds, following the two-loop structure of
\autoref{alg:tri_plain}. Each round processes a single vertex $i$ in the order
$\prec$ induced by the noisy degrees. The round computes all neighborhood
intersections involving $i$ while ensuring that the observable memory-access
pattern depends \emph{only} on its noisy degree $\hat d_i$, never on the true
graph structure. The core novelty of \sysname{} is twofold:
(1) a new access-pattern hiding mechanism that issues exactly \emph{one} fake
access per real access, avoiding the $\Omega(\log n)$ overhead of generic ORAM
by exploiting our algorithm's structure; and (2) a lightweight MPC procedure
that reduces no-ops for triangle enumeration by optimizing the
neighborhood-intersection logic.

\noindent \textbf{Fake Access Generation.}
For efficiency, \sysname{} performs inner-loop retrievals via secret-shared hash tables. While this enables $O(1)$ lookups, it creates a leakage channel. In round~1, the servers see only random accesses to $\Tu$, revealing nothing. From round~2 onward, repeated accesses can reveal structure: for example, re-accesses in round~2 expose how many neighbors vertices~1 and~2 share. Over rounds, such patterns can cumulatively leak the graph. One mitigation is to rebuild the hash tables each round, but with the table size $n\hat d_1$ this costs $O(n^2\hat d_1)$. Another is generic ORAM-style hiding, but it incurs an $\Omega(\log n)$ overhead per access. We include an experimental evaluation of both baselines in Sec.~\ref{sec:evaluation}. Instead, \sysname{} adopts a hybrid approach grounded in an important observation: the access pattern in \autoref{alg:tri_plain} is highly structured. Round $i$ performs exactly $\hat d_i$ lookups, where $\hat d_i$ is public. Thus, the leakage comes from whether
each lookup corresponds to a vertex whose neighbor list was already retrieved
in a previous round (``seen'') or is being accessed for the first time
(``unseen''). This seen/unseen imbalance directly reveals overlap in neighborhoods.

\update{Our \textbf{core idea} is to \textit{equalize} this view with only \emph{one} fake access per real access: in round $i$, the servers observe exactly (i) $\hat d_i$ accesses sampled uniformly from the \textbf{seen} set and (ii) $\hat d_i$ sampled uniformly from the \textbf{unseen} set, independent of vertex~$i$'s true neighborhood. To achieve this, we maintain two disjoint structures: (i) a growing \textbf{seen pool} $\Seen$, and (ii) a shrinking \textbf{unseen pool} $\US$, which initially contains all $n$ vertices. In principle, both pools could be implemented using dynamic distributed ORAMs. However, to the best of our knowledge, existing practically efficient distributed ORAM implementations \cite{doerner2017scaling,vadapalli2023duoram,falk2023gigadoram} do not support dynamic size in the malicious-security setting with two-servers. For our concrete instantiation, we therefore implement $\Seen$, as mentioned above, using an oblivious hash table $\Ts$ that is securely shuffled and rebuilt in every iteration. For the unseen pool, each vertex is accessed at most once before being promoted to $\Seen$, making recent OSAM constructions~\cite{heath20262pc,pia2026oblivious} a more efficient alternative to generic ORAM. However, available implementations~\cite{heath20262pc} support only fixed 128-bit records, which is unsuitable for our neighbor-lists. We therefore implement $\US$ using a custom oblivious data structure built around another oblivious hash table $\Tu$. We emphasize that these are \emph{implementation choices}, driven by the practical efficiency of currently available constructions. The algorithmic design of \sysname{} is \textit{agnostic} to how $\Seen$ and $\US$ are instantiated, allowing more efficient constructions--including future ORAM-based ones--to be substituted directly.}

$\Ts$ is rebuilt once per round to eliminate duplicate touches. Its size remains small: by round $i$ it contains only $\sum_{j \le i} \hat d_j$ vertices, so rebuilding is far cheaper than a full-table rebuild, and periodic refreshes further reduce amortized cost (see `Reset State''). In contrast, $\Tu$ can be reused across rounds since unseen vertices are never accessed twice. Together, these choices let \sysname{} support exactly one fake access per real access while hiding structural dependencies. The remaining challenge is to issue accesses that (i) generate $\hat d_i$ \textit{distinct} random seen and $\hat d_i$ \textit{distinct} random unseen vertices, (ii) include all real neighbors of $i$ for correctness, and (iii) make fake accesses indistinguishable from real ones.


\begin{figure}[t]
    \centering
    \includegraphics[width=1\linewidth]{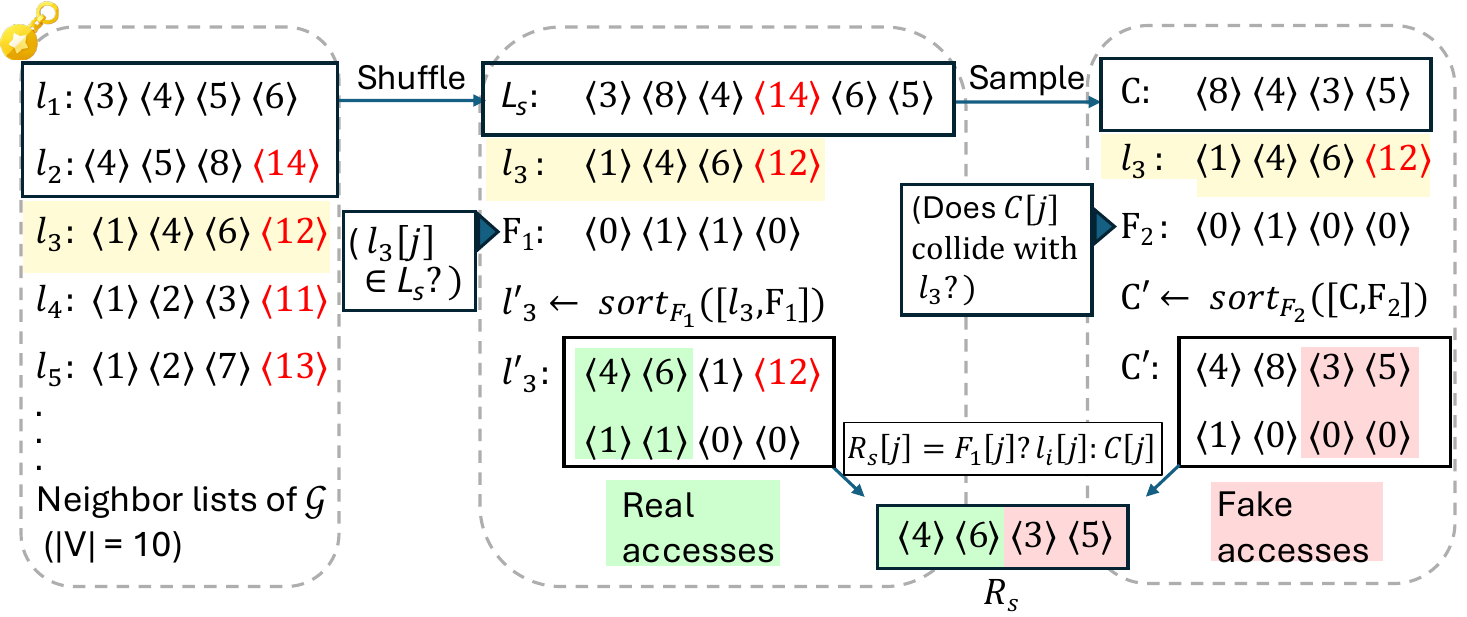}
    \caption{\sysname{}: \textbf{Seen request $\Rs$ generation}. For $\ell_3$, the protocol outputs real accesses $u\in\ell_3$ with $\F_1(u)=[u\in\Seen]$, and fake accesses using candidates $c\leftarrow\Seen$ with $\F_2(c)=[c\in\ell_3]$ set to $0$ (non-colliding), yielding the final mixed access sequence. The same procedure is applied for the unseen cache (e.g., for $1$ and $12$). $\langle$Red$\rangle$ means dummy nodes.}

    \label{fig:seen-access}
    \vspace{-5mm}
\end{figure}

\sysname{} achieves this via a custom MPC protocol. We explain the construction of the seen-access list $\Rs$ (Fig.~\ref{fig:seen-access}). Let $b$ denote the number of neighbors of $i$ that are already in the seen set $\Seen$. The goal is to construct $\Rs$ with exactly $\hat d_i$ entries: all $b$ real seen neighbors, plus $\hat d_i-b$ fake accesses.

The servers first securely sample $\hat d_i$ distinct candidate fake vertices from $\Seen$, denoted $\C$. They then compute two secret-shared indicator flags:
\begin{packeditemize}
\item $\F_1[j]$ indicates whether the $j$-th entry of $l_i$ is a real seen neighbor,
\item $\F_2[j]$ indicates whether the $j$-th candidate in $\C$ also appears in $l_i$, and would therefore duplicate a real neighbor.
\end{packeditemize}

Using these flags, the servers form $\N_i'[j]=(\F_1[j],l_i[j])$ and $\C'[j]=(\F_2[j],\C[j])$, and securely sort both lists by their flag bits in descending order. This moves the $b$ real seen neighbors to the front of $\N_i'$ and the colliding fake candidates to the front of $\C'$. The servers can therefore construct $\Rs$ by taking the real neighbors from the front of $\N_i'$ and filling the remaining $\hat d_i-b$ positions with non-colliding candidates from $\C'$ (i.e., $\Rs[j]\!=\!\F_1[j]\;?\;l_i[j]\!:\!\C[j]$).
Finally, a secure shuffle hides the positions of accesses before $\Rs$ is used for memory access. The unseen-access list is constructed analogously.
\\\noindent\textbf{Triangle Detection.}
After retrieving the $2\hat d_i$ neighbor lists for the round, the protocol computes the triangle contributions involving vertex~$i$. The core operation is to find common neighbors of $i$ and each $v \in l_i$. For enumeration, we output triples $(i,v,w)$ for each $w \in l_i \cap l_v$; for counting, we update only a running tally. To ensure correctness, \sysname{} maintains three classes of flags to suppress invalid contributions:
\\\noindent (i) \textit{Dummy Padding Flags.}
Neighbor lists are padded to the noisy degree, so some entries are dummy identifiers; these must never contribute to triangle counts.
\\\noindent (ii) \textit{Fake-access Flags.}
Both seen and unseen request lists include fake vertices used only to hide access patterns; any triple involving a fake vertex must be filtered out.
\begin{figure}
    \centering
    \includegraphics[width=1\linewidth]{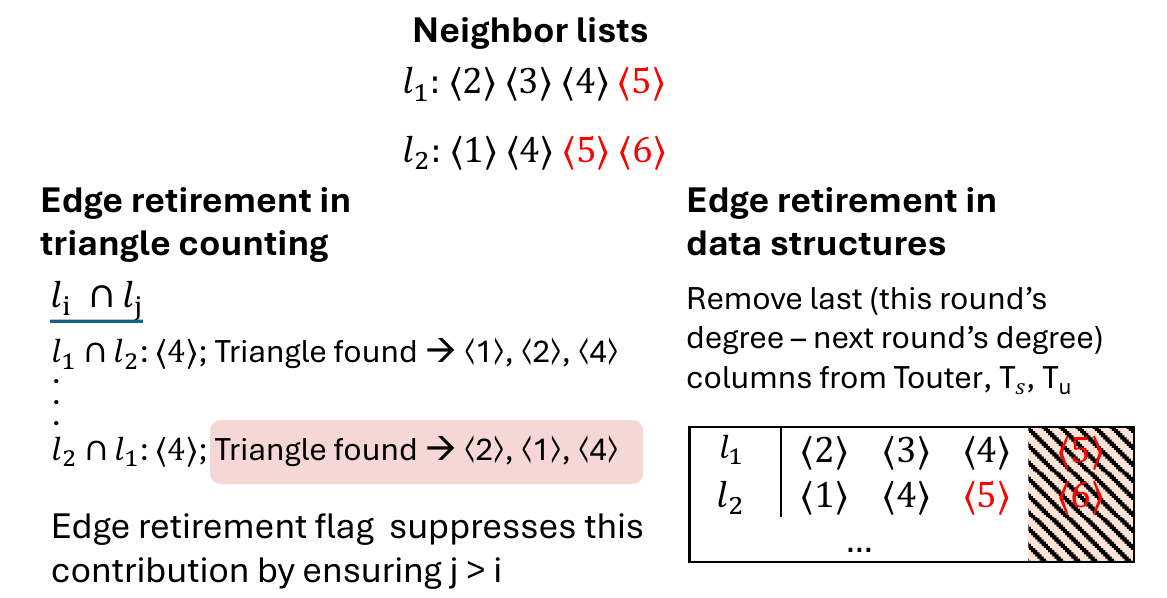}
    \caption{\sysname{}: Edge retirement. $\langle$Red$\rangle$ means dummy nodes.}
    \label{fig:edge-retirement}
    \vspace{-3mm}
\end{figure}
\\\noindent (iii) \textit{Edge Retirement Flags.}
All neighbor lists in $\Tu$ and $\Ts$ are initially padded to the global maximum noisy degree $\hat d_1$ for obliviousness. However, by round $i$, all triangle contributions involving vertex $i$ have been accounted for, allowing $i$ and its incident edges to be ``retired.'' The active padding can then decrease from $\hat d_i$ to $\hat d_{i+1}$. This requires (1) suppressing retired edge contributions and (2) safely removing the corresponding padding. \sysname{} addresses both using the degree ordering~$\prec$. Without retirement, an edge $(i,j)$ with $i<j$ could contribute in both rounds $i$ and $j$. Since vertices are re-indexed by $\prec$ and each neighbor list is sorted accordingly, vertex $v$ can appear only within the first $v$ positions of any neighbor list. This invariant lets \sysname{} construct a retirement flag that suppresses the later contribution of $(i,j)$ in round $j>i$ (see Fig.~\ref{fig:edge-retirement}). For padding, dummy IDs are $>n$ and therefore appear at the end of each sorted list, allowing the servers to remove the final $\hat d_i-\hat d_{i+1}$ elements before round $i+1$.
\begin{figure}
    \centering
    \includegraphics[width=1\linewidth]{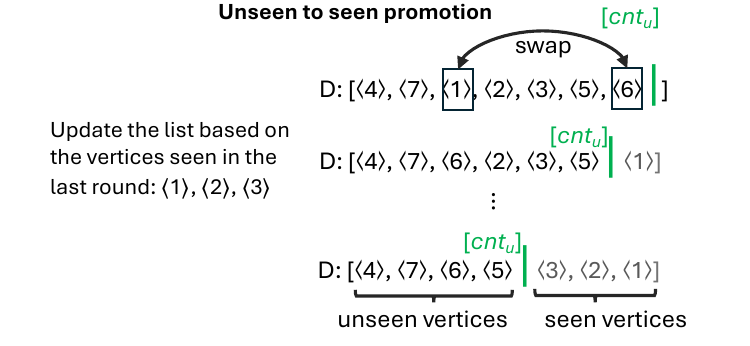}
    \caption{\sysname{}: Unseen to seen promotion in update state.}
    \label{fig:update-state}
    \vspace{-5mm}
\end{figure}

\noindent\textbf{Update State.}
\update{At the end of round~$i$, \sysname{} updates both states by obliviously migrating all vertices newly accessed in this round from the unseen structure $\US$ to the seen structure $\Seen$. Recall that $\US=(\Tu,D,D_u,\cnt_u)$, where the active prefix $D[1..\cnt_u]$ contains exactly the currently unseen vertices, while $\Seen=(\Ts,D_s,\cnt_s)$. Each entry in $\Tu$ stores the vertex's neighbor list together with a pointer to its current position in $D$. A na\"ive oblivious deletion from $D$ would require a linear scan. Instead, \sysname{} exploits its prefix representation to support $O(1)$ deletion via \emph{swap-delete}. For each newly seen vertex sampled from $D_u$ in round $i$, its lookup in $\Tu$ returns its position in $D$. The servers swap this entry with the last active entry $D[\cnt_u]$, decrement $\cnt_u$ and increment $\cnt_s$ (Fig.~\ref{fig:update-state}). After all newly seen vertices are removed, the active prefix $D[1..\cnt_u]$ forms the new unseen list, while the removed suffix forms the seen list. The servers securely shuffle both to obtain fresh $D_u$ and $D_s$ for the next round. Since subsequent access generation uses only these freshly shuffled lists, the updates reveal no additional information.

Finally, the servers rebuild the seen hash table $\Ts$ to hide repeated accesses to seen vertices. Crucially, the much larger unseen table $\Tu$ need not be rebuilt: each unseen vertex is accessed at most once before being promoted to $\Seen$ and future unseen candidates are sampled from the freshly shuffled $D_u$ rather than directly from $\Tu$. Thus, $\Tu$ can be safely reused across rounds as long as sufficiently many unseen vertices remain for fake-access generation (explained below), avoiding an expensive rebuild of the large unseen table every round.}
\\\noindent\textbf{Reset State.} Since we traverse vertices in decreasing noisy-degree order in the outer loop (i.e., $\hat{d}_i > \hat{d}_{i+1}$), we always have enough seen entries to support the required number of fake accesses. However, fake-unseen sampling continually consumes entries from the unseen pool, and eventually the pool may become too small.
When this occurs--specifically, when $\cnt_u < \hat d_i$--the servers perform a global \emph{reset}. During a reset, (1) $\US$ is restored to contain the full vertex universe: the hash table $\Tu$ is rebuilt over all vertices and $D$ is reset to the complete list of vertices, and (2) $\Seen$ is cleared.

\begin{tcolorbox}[
colback=gray!10!white, breakable,
colframe=black,
boxrule=0.6pt,
arc=1mm
]
\footnotesize
In the loop execution phase, \sysname{} introduces three key ideas that drive its performance gains:
\begin{itemize}[leftmargin=1.2em,itemsep=2pt]

\item First, \sysname{} introduces a new approach to access-pattern hiding tailored to graph computation: only one fake access per real access is enough to hide topology-dependent access patterns, replacing generic ORAM's $\Omega(\log n)$ per-access overhead with constant overhead.

\item Second, \sysname{} introduces customized oblivious data structures that support efficient updates while avoiding costly full-table rebuilds. In particular, a custom oblivious swap-delete mechanism enables $O(1)$ deletion of consumed unseen vertices.
\item \sysname{} further turns the publicly released noisy degrees into another powerful computational optimization: leveraging the induced public ordering $\prec$, \sysname{} introduces a novel technique to eliminate a significant fraction of no-ops, substantially reducing the overall cost.
\end{itemize}
\end{tcolorbox}
\begin{mybox2}{\algotitle{alg:c4_plain}{\textsc{QuadrangleDetect}}}
\footnotesize
\textbf{Input:} Undirected graph $\G=(V,E)$ with neighbor lists $\{\ell_v\}_{v\in V}$, $|V|=n$. \\
\textbf{Output:} List $\mathcal{L}$ of 4-cycles and count $\mathit{cnt}$.
\begin{enumerate}[nosep, label=\arabic*.]
\item $\mathcal{L} \gets [\,]$, \;\; $\mathit{cnt} \gets 0$
\item Sort vertices as $v_1,\ldots,v_n$ such that $d_{v_1}\ge \cdots \ge d_{v_n}$ (break ties by id), and relabel them so that $v_i$ has id $i$.
\item \textbf{For} $i\in[n]$:
\\ \hspace*{1.2em} Initialize $U[w]\gets \emptyset$ for all $w\in V$
\\ \hspace*{1.2em} \textbf{For each} $u \in \N_i$ with $u>i$:
\\ \hspace*{2.4em} \textbf{For each} $w \in \N_u$ with $w>i$:
\\ \hspace*{3em} $U[w] \gets U[w]\cup\{u\}$
\\ \hspace*{3em} \textbf{If} $|U[w]|\ge 2$:
\\ \hspace*{4em} $\mathcal{L}.\text{Append}(i,w,U[w])$ \hfill \textcolor{blue}{$\rhd$ encodes cycles $(i,u_1,w,u_2)$, with $u_1, u_2 \in U[w]$}
\\ \hspace*{4em} $\mathit{cnt} \gets \mathit{cnt} + \binom{|U[w]|}{2}$
\item \Return $(\mathcal{L},\mathit{cnt})$
\end{enumerate}
\end{mybox2}

\section{\sysname{}: Quadrangle Detection}
\label{sec:quadrangle}

\subsection{Technical Overview}
As before, we begin with an intuitive explanation of the algorithm and then describe
how it is translated into MPC. For quadrangle detection, the key idea is to reduce the problem to identifying
vertex pairs that share at least two distinct common neighbors. Fix a pair
$(i, w)$. A quadrangle $(i, u_1, w, u_2)$ exists exactly when both $i$
and $w$ share two distinct neighbors $u_1, u_2 \in \N_v \cap \N_w$. 
Equivalently,
$w$ participates in a quadrangle with $i$ when it is reachable as a two-hop neighbor of
$i$ via at least two different intermediate vertices in $l_i$. Alg.~\ref{alg:c4_plain} implements this by iterating over each vertex $v$ and exploring all two-hop paths $i \to u \to w$ with $u \in \N_v$ and $w \in \N_u$. The algorithm maintains an auxiliary structure $U[\cdot]$ where $U[w]$ collects all intermediates $u$ that realize such two-hop paths. Once $|U[w]| \ge 2$, we have certified the presence of a quadrangle involving $(v,w)$. The algorithm then emits a record $(i,w,U[w])$, from which every unordered pair of distinct vertices in $U[w]$ yields a quadrangle $(i,u_1,w,u_2)$.


\begin{figure}[t]
    \centering
    \includegraphics[width=1\linewidth]{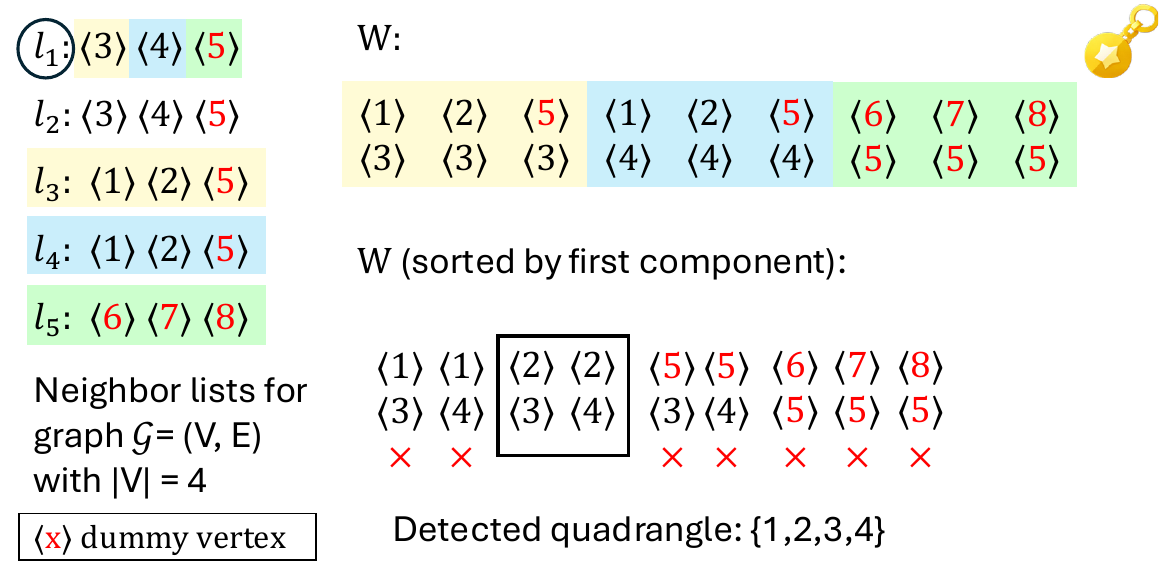}
   \caption{\sysname{}: quadrangle detection (example). For vertex $1$, we iterate over neighbors ($3,4,5$), fetch $\ell_3,\ell_4,\ell_5$, and append into $W$: $(u,3)$ for each $u\in\ell_3$, then $(u,4)$ for each $u\in\ell_4$, and so on-- encoding paths $1\!\to\!3\!\to\!u$, $1\!\to\!4\!\to\!u$, and $1\!\to\!5\!\to\!u$. We sort $W$, compute run lengths, ignoring dummy (red) vertices and any $u\le 1$ to prevent double counting.}
   \vspace{-4mm}
    \label{fig:W}
\end{figure}

\subsection{Protocols}
At a high level, the core technical contribution is the construction of a custom MPC-friendly, data structure that captures the essence of \autoref{alg:c4_plain}. Quadrangle detection follows the same two-loop structure as triangle detection, differing only in the per-round detection logic. The three phases are unchanged: \textit{data collection} (Sec.~\ref{sec:datacollection}) is identical; the one-time \textit{setup} (Sec.~\ref{sec:protocol:setup}) runs the same well-formedness checks and initializes global state, with one additional structure $\W$; and \textit{loop execution} (Sec.~\ref{sec:protocol:round}) keeps the same workflow. The outer loop processes vertices in decreasing noisy-degree order, where $i$ denotes the vertex handled in round $i$ (after re-indexing) and $\hat d_i$ its noisy degree. The inner loop retrieves neighbors of $i$ using the same fake access technique.

The key difference is how retrieved neighbors are used to detect cycles. Quadrangle detection must find, for each pair $(i,w)$, all intermediates $u$ forming two-hop paths $(i,u,w)$. The cleartext algorithm (\autoref{alg:c4_plain}) uses a position-indexed table $U[w]$, but directly implementing $U$ in MPC would leak access patterns through position-dependent inserts. \sysname{} instead uses an \textit{append-only} secret-shared list $\W$. In round $i$, for each $w \in l_i$, we append all tuples $\{(w,u)\mid u \in l_w\}$ to $\W$, then securely sort $\W$ by endpoint $w$. After sorting, a maximal run of $w$ of length $k$ represents $k$ two-hop paths from $i$ and contributes $\binom{k}{2}$ quadrangles $(i,u_1,w,u_2)$. We apply the same filtering flags as for triangles (dummy-padding, fake-access, and edge-retire flags) and additionally drop tuples $(i,\cdot)$, since $w=i$ yields a degenerate self-loop. Fig.~\ref{fig:W} shows an example. For counting, we store only $w$ rather than $(w,u)$ and sum $\binom{k}{2}$ over all maximal runs.

\vspace{-1mm}
\section{Implementation Optimizations}\label{sec:optimizations}
\noindent\textbf{Cheap Non-Linear Operations using FSS.}
\sysname{} frequently uses comparisons $(\geq,\leq,=,\neq)$ over secret-shares, which is expensive over MPC because comparisons require multiple rounds of communication. To resolve this, we use function secret sharing, yielding an essentially one-round online cost, similar to MPC with FSS gates~\cite{cryptoeprint:2020/1392,cryptoeprint:2019/1095,wagh2022pika}. 

Consider testing whether $x \ge y$ for secret-shared $x,y \in [k]$, and let $z=x-y$. Suppose we can convert $\langle z\rangle$ into shares of the DPF $f_{z,1}$ over the domain $[\B=2k]$, where indices $0$ to $k-1$ represent non-negative integers and $k$ to $2k-1$ represent negative integers under modular wrap-around. Expanding this DPF yields a one-hot vector $\bold{e}_{z,1}$ with $\bold{e}_{z,1}[i]=1$ exactly at $i=z \bmod \B$. Then checking $z \ge 0$ reduces to summing $\bold{e}_{z,1}$ over the non-negative half of the domain; equivalently,
$
  x < y \;\Longleftrightarrow\; \sum_{i=0}^{k-1} \bold{e}_{z,1}[i] = 0.
$
Other comparisons ($x>y$, $x=y$, $x\neq y$) follow similarly. The remaining challenge is efficiently converting a secret-shared value $\langle z \rangle$ into shares of its corresponding DPF. To do so, we use a preprocessed correlated-randomness setup via \emph{verifiable FSS} (VFSS), which lets the parties check that the expanded vector is a valid one-hot encoding. The dealer samples a random value $r \in \mathbb{\F}$ and provides VFSS keys for the DPF $f_{r',1}$ with $r' = r \bmod 2n$, along with secret shares of $r$. Because \sysname{} uses SPDZ-style MACs (Sec.~\ref{sec:background}), the expanded one-hot values must also be authenticated. Since a DPF encodes a point function, its MAC shares also form a DPF, so the dealer additionally provides VFSS keys for $f_{r',\alpha}$, where $\alpha$ is the global MAC key. As noted in Sec.~\ref{sec:setting}, this is natural since the same dealer also provides the global MAC key for the rest of the computation. The servers evaluate VFSS keys for $f_{r',1}$ over the full domain to obtain the expanded one-hot vector $\bold{e}_{r',1}$. In the online phase, they open $a = (z - r) \bmod 2n$ and use it to ``shift'' the preprocessed vector $\bold{e}_{r',1}$, thereby obtaining DPF shares of $z$ via
$
  \bold{e}_{z,1}[j] = \bold{e}_{r',1}[j - a].
$
The same cyclic shift is applied to the MAC-DPF $f_{r',\alpha}$, ensuring the resulting shares remain authenticated. In the malicious setting, we also verify that the dealer provided well-formed VFSS keys and consistent shares of $r$, via a lightweight malicious-sketching protocol. Importantly, evaluating VFSS keys, expanding into one-hot vectors, and checking well-formedness all occur in preprocessing. Thus, the online cost of each comparison is extremely small: a single opening plus a few local operations (full details in App.~\ref{app:FSS}).
\\\noindent\textbf{Parallelization.}
The dominant cost in our protocol is cycle detection, for both triangles and quadrangles, and this work is highly parallelizable. Recall from Sec.~\ref{sec:protocol:round} that in round $i$, the servers first retrieve the $2\hat d_i$ neighbor lists (real and fake), then compute cycle contributions involving $i$. Crucially, after retrieval, detection depends only on the lists retrieved in that round. We therefore decouple execution into two stages: (1) sequentially perform the oblivious retrievals for each $i$ and record the retrieved lists; and (2) compute cycle contributions from these lists. Stage (2) is embarrassingly parallel, as each $i$ can be processed independently (App.~\ref{app:parallel}).

\section{Security Analysis}\label{sec:security}
We use the simulation paradigm~\cite{Oded} to establish \sysname's security. This paradigm defines two worlds: a real execution among the parties, and an ideal execution where a functionality $\mathcal{F}$ receives all inputs and directly computes the output. A simulator $\mathsf{Sim}$ interacts with $\mathcal{A}$ and $\mathcal{F}$ to generate $\mathcal{A}$'s view. If $\mathcal{A}$ cannot distinguish the real and ideal executions, the protocol securely realizes $\mathcal{F}$ and leaks nothing beyond what $\mathcal{F}$ permits.

The ideal functionality for \sysname{} is shown in Fig.~\ref{fig:IF}. A key distinction from Matthew Groce et al.~\cite{Groce2019CheaperPS} is that the DP mechanism is not executed inside MPC or revealed via the output. Instead, it occurs before the protocol: each data-owning party inputs its noisy degrees directly into MPC. Thus, any leakage arises only from these DP-sanitized inputs, while the MPC reveals nothing beyond the designated outputs. Accordingly, the ideal functionality treats $\hat d_v$ as part of each party's input.

\begin{mybox2}{Ideal Functionality}
\footnotesize
\textbf{Parameter.} $\epsilon, \delta$ - Privacy parameters
\\\textbf{Parties.} $m$ Data owners $P_j, j \in [m]$ holding partition of vertex sets $V_j$ and the incident edges $E_j$; Servers $\Ser_0$ and $\Ser_1$ 
\\
\textbf{Functionality.}
\begin{packeditemize}
\item Each honest party $P_j$ samples $\eta_v \sim \text{TSDLap}\!\left(\tfrac{2}{\epsilon},\, \lceil 2 + \tfrac{2}{\epsilon}\log(2/\delta) \rceil\right)$ and computes $\hat{d}_v=d_v+\eta_v$; corrupted parties receive values $\hat{d}_v$ from the adversary
\item Receive from party $P_j$ its neighbor lists and the values $(\ell_v, \hat{d}_v)$ for all $v \in V_j$
\item Let $V=\bigcup_{j=1}^m V_j$, $E=\bigcup_{j=1}^m E_j$ $n=|V|$ and $t=\lceil 2+2/\epsilon\log(2/\delta)\rceil$
\item If $\exists v \in V$ s.t. $\hat{d}_v \not \in [d_v, d_v+2t]$, send ABORT 
\item Construct the adjacency matrix $A$ of the entire graph over $V$ using the received neighbor lists and checks if $A_{jk}=A_{kj}, \forall j,k \in [n]$, send ABORT 
\item Run \autoref{alg:tri_plain} on the entire graph $\G=(V,E)$ to obtain $(cnt_{T}, \mathcal{L}_{T})$ for triangles and run \autoref{alg:c4_plain} on $\G$ to obtain  $(cnt_Q, \mathcal{L}_Q)$
\item Shuffles $\mathcal{L}_T$ and $\mathcal{L}_Q$
\item Sends $cnt_{T}, cnt_{Q}$ (or the lists $\mathcal{L}_Q$, $\mathcal{L}_T$) to the servers
\end{packeditemize}
\end{mybox2}
 \noindent\begin{minipage}{\columnwidth}
\captionof{figure}{Ideal functionality of \sysname.{}}\label{fig:IF}
\end{minipage}

Although we operate in the malicious threat model, we do not require extra data-owner checks to verify correct noise sampling. In our distributed setting we effectively use the \emph{local DP} model: each party adds noise independently, and every honest party $P_j$ obtains an $(\epsilon,\delta)$-edge local DP guarantee from $\mathsf{TSDLap}$ for its subgraph $G_j$, regardless of malicious behavior. Thus, malicious parties \emph{cannot} weaken the privacy of honest parties. Intuitively, parties release noisy degrees, then submit secret-shared, dummy-padded neighbor lists consistent with these degrees as input to MPC. The MPC computation reveals nothing beyond what is already implied by the noisy degrees and the final cycle counts/lists.\\

Coming to utility (accuracy), note that the dummies added for DP do not \textit{affect} triangle or quadrangle computation, since they are filtered out in MPC. A malicious party could try to degrade utility by falsifying an edge shared with an honest party, but our edge-consistency checks prevent this (Sec.~\ref{sec:protocol:setup}). The main remaining attack is to claim an extremely large noisy degree (e.g., $n$), inflating cost via no-ops, but our bounded-dummy-padding rule limits any claimed noisy value to a verifiable range implied by the noise distribution. Formally:

\begin{theorem}
\sysname~achieves the ideal functionality in Fig. \ref{fig:IF} with $(\epsilon, \delta)$-edge local DP leakage in the presence of a static, malicious adversary
corrupting one server and an arbitrary number of data owing parties in the $(\mathcal{F}_{\textsf{Sort}},\mathcal{F}_{\textsf{Shuffle}}, \mathcal{F}_{\textsf{Rand}})$-hybrid model (proof is in App. \ref{app:security}). \label{thm:security}
\end{theorem}

\noindent\textbf{Complexity Analysis.}\label{sec:complexity}
\update{Since all loop bounds in \sysname{} are determined by the public noisy degrees, its work decomposes into (i) access generation, (ii) cycle detection, and (iii) state maintenance. For sparse graphs\footnote{For sparse graphs where $\sum_i\hat d_i=O(n)$. For general graphs, request generation contributes an additional $O(n\sum_i\hat d_i)$ term.}, access generation and state maintenance each cost $O(n^2)$. Detection costs $\Theta(\sum_i \hat d_i^3)$ for triangle and $\Theta(\sum_i \hat d_i^2\log\hat d_i)$ for quadrangle detection. Thus, the overall costs are $O(n^2+\sum_i\hat d_i^3)$ for triangles and $O(n^2+\sum_i\hat d_i^2\log\hat d_i)$ for quadrangles. Full derivation in App.~\ref{app:complexity}.}

\section{Differentially Private Cycle Counts}
\label{sec:DP}
For the \emph{counting} variant, we provide an end-to-end edge-DP guarantee. After computing the exact cycle count under MPC, the servers jointly sample Laplace noise within the protocol and add it to the output~\cite{cryptoeprint:2019/823,Fu_2024}. Thus, the only revealed information, the noisy degrees and final noisy count, is protected under edge DP. Releasing noisy degrees also improves accuracy: the global sensitivity of cycle counting depends on the maximum degree $d_{\max}$. Without a DP estimate, one must assume the worst case ($d_{\max}=n$), resulting in excessive noise. In contrast, \sysname{} uses the released noisy degrees to obtain a tighter data-dependent bound, substantially reducing the required noise. The overall privacy loss composes across the degree release and final output.

Our protocol thus fits naturally into \emph{distributed differential privacy}~\cite{BohlerK20,Kerschbaum,CryptE,Bell22,keller2025piquantvarepsilonprivatequantileestimation}. The local model avoids a trusted server but incurs substantially worse accuracy, while the central model offers better utility but assumes a trusted curator. For example, for triangle counting, central DP achieves $\ell_2$ error $O(d_{\max}^2/\epsilon^2)$, whereas local DP incurs
$O!\left(
\frac{e}{\epsilon (e^{\epsilon}-1)^2}
\left(
d_{\max}^3 n
+
\frac{e}{\epsilon^2} d_{\max}^2 n
\right)
\right)$.
Distributed DP bridges this gap using MPC, achieving central-DP accuracy without a trusted server. \sysname{} instantiates this approach for distributed triangle and quadrangle counting.

One subtlety is that end-to-end DP requires verifying that data owners sample noise correctly, since the guarantee is in the central DP model (unlike the per-party local DP guarantee in Sec.~\ref{sec:security}). This can be enforced by zero-knowledge proofs of correct noise sampling, using techniques from~\cite{narayan2015verifiable,biswas2022verifiable,cryptoeprint:2025/851}.

\begin{figure*}[t]
\centering
\includegraphics[width=.23\textwidth]{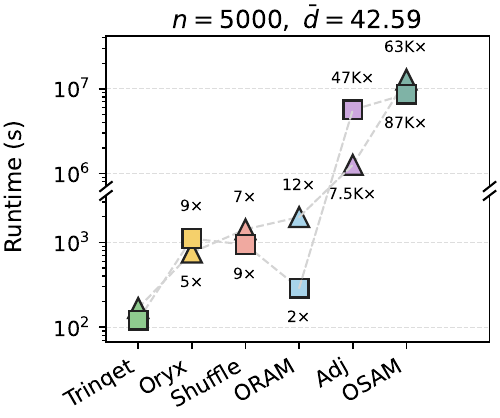}\hfill
\includegraphics[width=.22\textwidth]{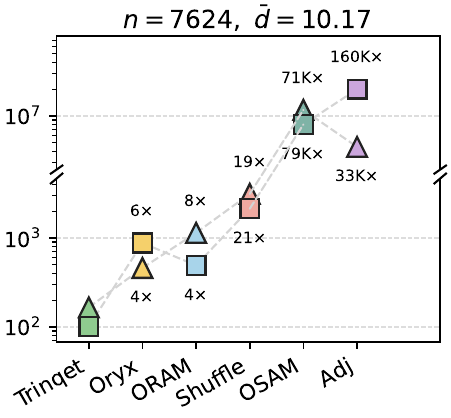}\hfill
\includegraphics[width=.22\textwidth]{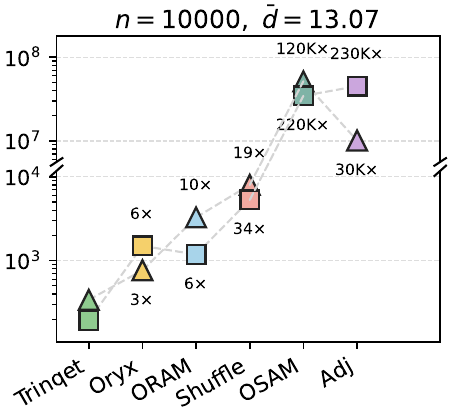}\hfill
\includegraphics[width=.20\textwidth]{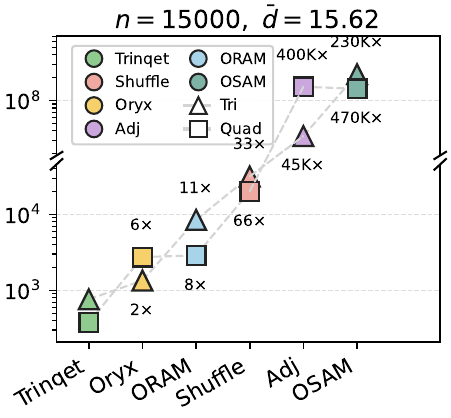}

\includegraphics[width=.23\textwidth]{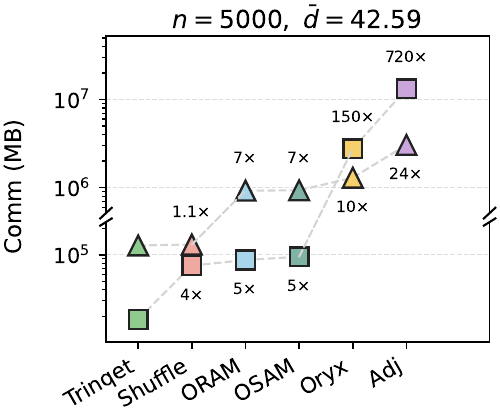}\hfill
\includegraphics[width=.23\textwidth]{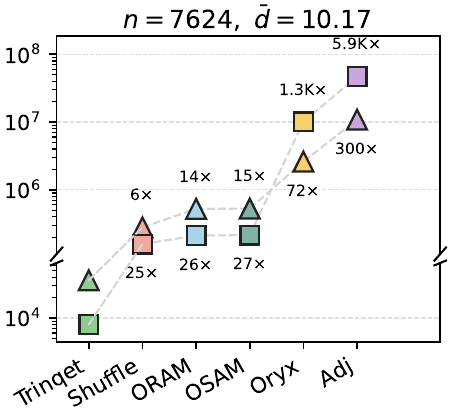}\hfill
\includegraphics[width=.23\textwidth]{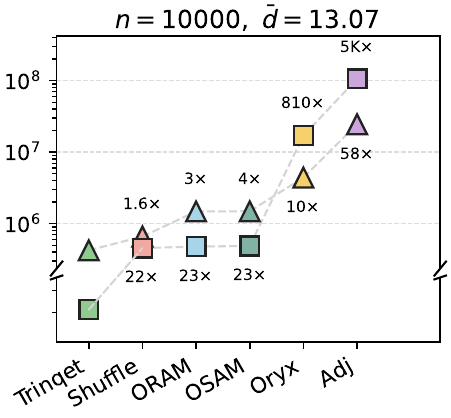}\hfill
\includegraphics[width=.22\textwidth]{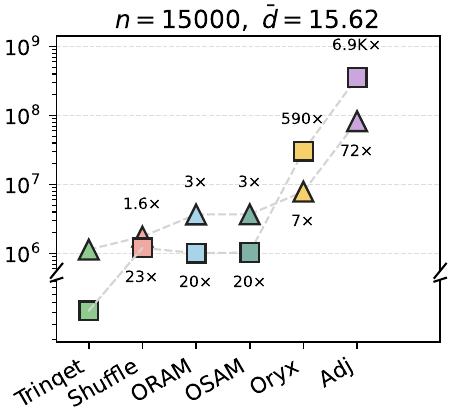}

\caption{Runtime (top row) and total communication (bottom row) of \sysname{} vs.\ baselines on four datasets (columns, left to right). \textit{ORAM}, \textit{OSAM}, and \textit{Oryx}~\cite{zhong2024oryx} operate in a weaker threat model than \sysname{} (semi-honest vs.\ malicious).}
\vspace{-5mm}
\label{fig:baselines}
\end{figure*}

\begin{figure}[t]
\centering
\includegraphics[width=1\linewidth]{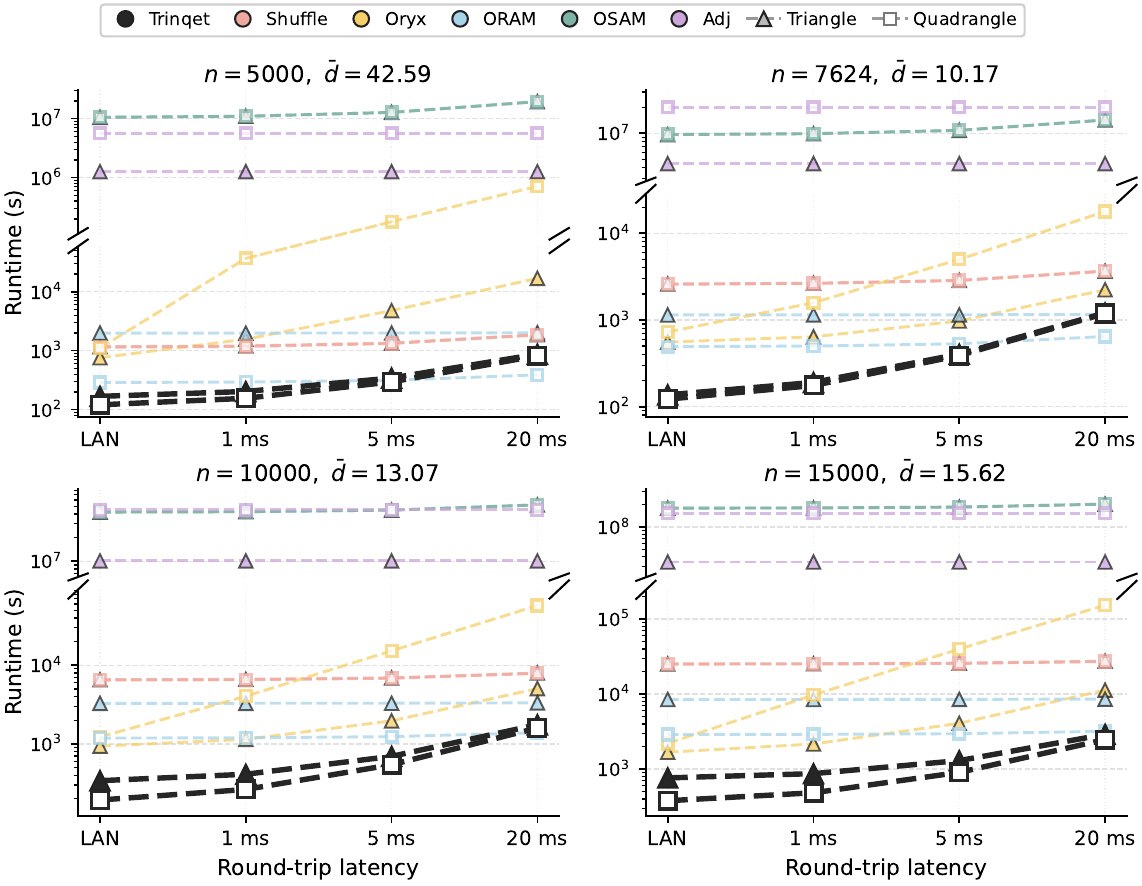}
\caption{Runtime under LAN, 1\,ms, 5\,ms, and 20\,ms RTT. }
\vspace{-5mm}
\label{fig:network-latency}
\end{figure}

\begin{figure}[t]
\centering
\includegraphics[width=.47\linewidth]{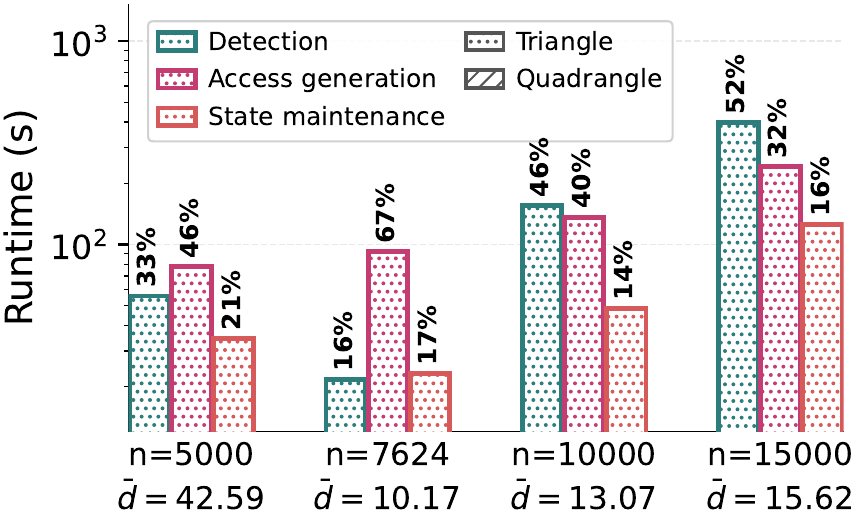}\hfill
\includegraphics[width=.47\linewidth]{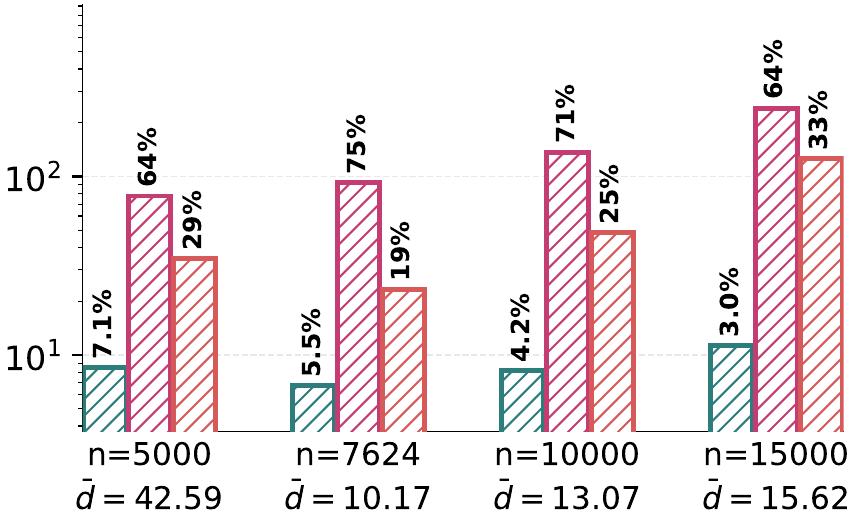}


\includegraphics[width=.47\linewidth]{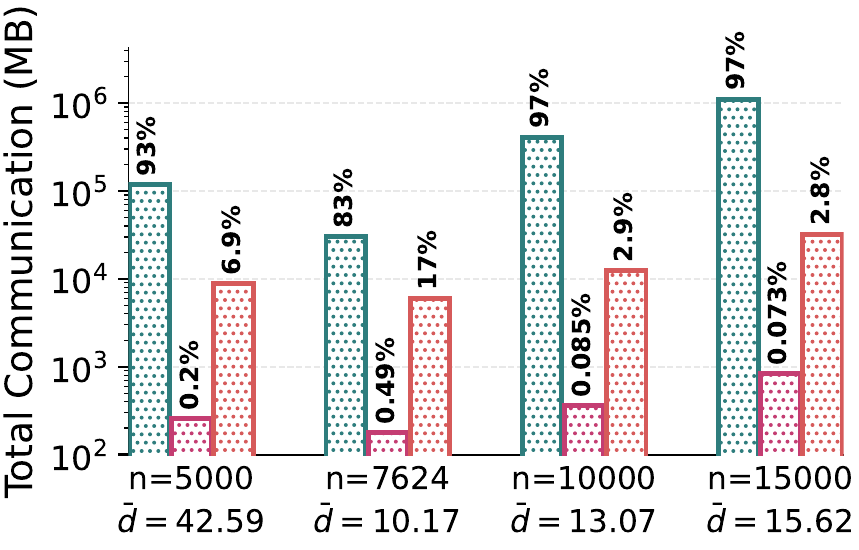}\hfill
\includegraphics[width=.47\linewidth]{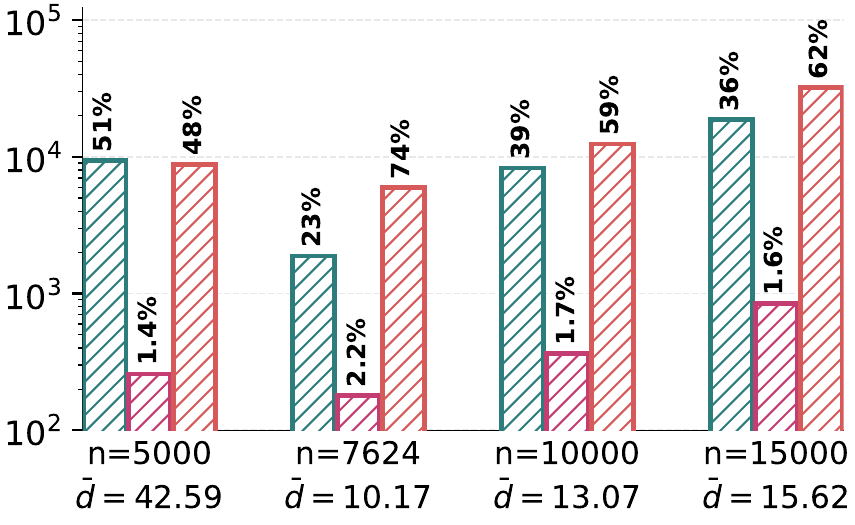}

\caption{Per-component cost breakdown of \sysname\ (\%). The top and bottom rows show runtime and communication; the left and right columns show triangle and quadrangle costs.}
\vspace{-5mm}
\label{fig:breakdown}
\end{figure}

\section{Evaluation}\label{sec:evaluation}
\noindent\textbf{Setup.} We implement \sysname{} in MP-SPDZ~\cite{keller2020mp} using the oblivious shuffle of~\cite{cryptoeprint:2025/2137}, on machines with 48-core Intel Xeon CPUs and 125GB RAM. We set $(\varepsilon,\delta)=(1,10^{-8})$ for noisy degree release and evaluate on RTTs of 1, 5, and 20 ms.



\noindent \textbf{Datasets.}
We evaluate using four undirected graphs and denote them by \Ds1--\Ds4. Additionally, $n$ denotes the number of vertices, $\bar d$ the average degree, and $d_{\max}$ the maximum degree.
Two are real-world datasets: an \textsc{SBM} instance from the Stochastic Block Partitioning Challenge~\cite{graphMIT} ($n{=}5000$, $\bar d{=}42.59$, $d_{\max}{=}174$) and a social network dataset of LastFM users, \textsc{Feather}~\cite{feather} ($n{=}7624$, $\bar d{=}10.17$, $d_{\max}{=}217$).
We also generate two synthetic graphs with $n{=}10000$ ($\bar d{=}13.07$, $d_{\max}{=}281$) and $n{=}15000$ ($\bar d{=}15.62$, $d_{\max}{=}331$); details of the generation approach are given in App.~\ref{app:synthetic_dataset}.

\noindent \textbf{Baselines.} We compare \sysname{} against five baselines.
(1) \textit{Oryx}~\cite{zhong2024oryx} is a cycle-enumeration protocol that operates under a weaker semi-honest threat model than \sysname{} and additionally leaks all $k$-hop paths. It requires three online servers and assumes a bounded maximum degree. Designed for low-degree graphs, it exceeds memory on our higher-degree datasets; in these cases, we run Oryx with the largest feasible tuple budget and extrapolate using its provided cost model.
(2) \textit{Adjacency matrix-based baseline} (\textit{Adj}) constructs an $n\times n$ adjacency matrix, enforces undirected consistency by checking $A_{u,v}=A_{v,u}$, then enumerates all ${n\choose3}$ triples for triangles and all 4-tuples for quadrangles. Unlike Oryx, it matches our malicious threat model, providing a direct comparison.
\update{(3) \textit{ORAM} (semi-honest) follows the same high-level structure as \sysname{}, but replaces our fake-access mechanism with distributed ORAM. Each adjacency list is stored as a fixed $d_{\max}$-word block in a secret-shared block ORAM, instantiated using the DPF-based 3-party DuoRAM~\cite{vadapalli2023duoram}. The protocol requires $\sum_{i=1}^{n}\hat d_i$ block reads, parallelized across cores. Each read retrieves a $d_{\max}$-word record regardless of the vertex's actual degree. In contrast, \sysname{} uses public noisy degrees to operate on $\hat d_i$-length records; \textit{ORAM} therefore pays the global-$d_{\max}$ block cost per access. It also provides only semi-honest security.
(4) \textit{OSAM} uses recent Oblivious Single-Access Machines~\cite{heath20262pc,pia2026oblivious}, which reduce generic ORAM overhead by exploiting single-access memory. Since cycle detection may repeatedly access adjacency lists, our workload does not naturally satisfy this property. Nevertheless, we include this baseline, and use the RAM-over-SAM construction of SAM-2PC~\cite{heath20262pc}, which supports repeated logical accesses over OSAM (OSAM+~\cite{pia2026oblivious} instead targets the client--server setting). SAM-2PC supports only word-level access, with one oblivious tree traversal per 128-bit word and no block-read primitive. Packing two 64-bit vertex IDs per word, each $d_{\max}$-padded list occupies $b=\lceil d_{\max}/2\rceil$ words in a memory of $nb$ words. Thus, retrieving one list requires $b$ separate OSAM accesses rather than one block access. This is a limitation of the available implementation, not the OSAM abstraction. SAM-2PC provides semi-honest security.}
(5) \textit{Shuffle-based baseline} (\textit{Shuffle}) also follows \sysname{}'s high-level structure, but instead of fake accesses, obliviously shuffles all neighbor lists and rebuilds $\Tu$ every round. The last three baselines highlight the benefits of our fake-access design. We select these baselines because they support exact triangle counting; works such as CARGO~\cite{liu2024cargo} instead release noisy counts over projected graphs (App.~\ref{app:related-work}). Nevertheless, CARGO reports 485\,s on the Facebook graph~\cite{ego-facebook}, while \sysname{} achieves substantially lower runtime on much larger graphs.

\noindent \textbf{Performance Comparison.} Fig.~\ref{fig:baselines} reports end-to-end LAN runtime and communication. Most notably, on \Ds2 ($n{=}7624$), \sysname{} takes $<$2.3 minutes for detecting triangles and $<$2.1 minutes for quadrangles. \sysname{} achieves the best runtime across all five baselines, despite operating in a stronger threat model than \textit{ORAM}, \textit{OSAM}, and Oryx. Against the generic memory-hiding approaches, \sysname{} is up to $11.8\times$ (\Ds1) and $7.6\times$ (\Ds4) faster than \textit{ORAM} for triangles and quadrangles respectively, and up to $2.3{\times}10^{5}\times$ and $4.7{\times}10^{5}\times$ faster than \textit{OSAM} (on \Ds4). Among the baselines matching its malicious threat model, it is up to $33\times$ and $66\times$ faster than \textit{Shuffle} (on \Ds4). It is also up to $4.5\times$ and $9.2\times$ faster than Oryx (on \Ds1), even though Oryx is semi-honest and leaks all $k$-hop paths. The gains over the dense \textit{Adj} baseline are larger still, reaching $4.5{\times}10^{4}\times$ for triangles and $4.0{\times}10^{5}\times$ for quadrangles on \Ds4. Communication follows suit: \sysname{} reduces communication across all baselines, by up to $300\times$ for triangles (\Ds2) and $6.9{\times}10^{3}\times$ for quadrangles (\Ds4).

Beyond these results, we observe that (1) unlike \textit{Adj} and \textit{Oryx}, \sysname{} is cheaper for quadrangles than for triangles, which aligns with our complexity bounds since triangle detection scales as $\sum_i \hat d_i^{3}$ while quadrangle detection scales as $\sum_i \hat d_i^{2}\log \hat d_i$ (App.~\ref{app:complexity}); (2) \textit{Shuffle} highlights the higher cost of rebuilding $\Tu$ from scratch every round, whereas \sysname{} rebuilds only the smaller $\Ts$; (3) the relative gains vary predictably with $n$: gains over \textit{Adj} grow with graph size because its work is independent of sparsity; (4) \textit{ORAM}'s triangle runtime is 19.1 minutes (\Ds2), compared with 2.29 minutes for \sysname{}, because every block read pays for $d_{\max}$ entries even though most vertices have much smaller degree. 

\noindent \textbf{Network Latency.} \update{Fig.~\ref{fig:network-latency} evaluates all methods at 1ms, 5ms, and 20ms RTT. At 1\,ms and 5\,ms, \sysname{} achieves the lowest runtime across all four datasets. At 20\,ms, it remains fastest in five of the eight task--dataset combinations; \textit{ORAM} is faster for quadrangles on \Ds1--\Ds3 due to its lower round complexity. However, \textit{ORAM} operates under a weaker threat model, while \sysname{} remains fastest for triangles on the larger graphs (\Ds3 and \Ds4) and for quadrangles on \Ds4. Although \sysname{}'s runtime grows faster with RTT, it remains orders of magnitude faster than all other methods under realistic RTT regimes.}

\noindent \textbf{Cost Breakdown.} Fig.~\ref{fig:breakdown} shows the measured costs of different components of our protocol. Across data collection, setup, and loop execution, the dominant costs arise from cycle detection, access generation, and state maintenance during loop execution; the one-time setup cost is included under state maintenance. The breakdown highlights three observations consistent with our complexity analysis (App.~\ref{app:complexity}). \textbf{(1) Detection and access both drive triangle costs.} Triangle detection accounts for $16$--$52\%$ of runtime, while access generation accounts for $32$--$67\%$, reflecting the significance of both the $O(n^2)$ access cost and the $\Theta(\sum_i\hat d_i^3)$ detection cost. \textbf{(2) Lower quadrangle-detection complexity shifts the bottleneck to access generation.} For quadrangle counting, the lower $\Theta(\sum_i\hat d_i^2\log\hat d_i)$ detection cost accounts for only $3$--$7\%$ of runtime, leaving the $O(n^2)$ access-generation term as the dominant component at $64$--$75\%$. State maintenance remains secondary in runtime for both tasks, consistent with its lower $O(n^2)$ total-work bound. Thus, the contrasting triangle and quadrangle breakdowns empirically reflect their different degree-sensitive detection costs. \textbf{(3) Parallelism reduces the wall-clock impact of detection.} Triangle detection accounts for $83$--$97\%$ of communication despite only $16$--$52\%$ of runtime, showing independent operations parallelize effectively. 

\vspace{-1mm}

\section{Related Work}\label{sec:relatedwork}
A line of prior work~\cite{nayak2015graphsc,kapoor2025mathsf,graphiti,araki2021secure,mazloom2018secure,mazloom2020secure,peng2024mapcomp,zou2025ringsg,zou2024cognn} explores secure graph analytics, but does not target structural computations such as cycle counting. Most closely related to our use of DP leakage is Mazloom et al.~\cite{mazloom2018secure,mazloom2020secure} where servers themselves insert dummy edges which are explicitly routed to every node. This is different from our setting, as emulating this in \sysname{} would require data-owning parties to coordinate and consistently embed dummy edges in their neighbor lists, a prohibitively expensive approach. Our dummy vertices are merely placeholders and require no coordination among data owners. Oryx~\cite{zhong2024oryx} and Vorstermans~\cite{vorstermans2023secure} study private cycle detection under different paradigms: Vorstermans~\cite{vorstermans2023secure} uses costly adjacency-matrix MPC, while Oryx uses 3-server path enumeration under different leakage and threat assumptions.

Another line of work, CARGO~\cite{liu2024cargo} and MAGO~\cite{wang2023mago}, studies triangle counting under distributed DP. CARGO also leverages DP degrees, but for a different purpose: it uses them to \emph{bound sensitivity}, reducing the noise needed to release output. Its computation remains adjacency-matrix based ($O(n^3)$ worst-case). In contrast, \sysname{} uses noisy degrees only to \emph{reduce computation}: loop bounds can scale with graph sparsity. Thus, CARGO releases a noisy count over a projected graph, whereas \sysname{} performs exact cycle detection on the original graph and uses DP leakage for efficiency. MAGO instead uses join-style counting over uniformly padded neighbor views, with cost scaling with a global padded bound rather than per-node degrees. See App.~\ref{app:related-work} for a full classification.


 \section{Conclusion}\label{sec:conclusion}
 We presented \sysname{}, a scalable system for maliciously secure triangle and quadrangle detection over distributed graphs. By combining differentially private noisy degrees with sparsity-aware, access-pattern-hiding MPC, \sysname{} achieves strong privacy with high efficiency. Our evaluation shows that it substantially outperforms prior work, establishing a new design point for private graph analytics.

\section*{Ethical Considerations}
This paper studies privacy-preserving triangle and quadrangle detection on distributed graphs using MPC with bounded differentially private (DP) leakage. Graph edges can encode sensitive relationships, so we consider the ethical impact of both our methodology and publication.

\noindent\textbf{Stakeholders.}
Stakeholders include data-owning parties, end users whose relationships appear as edges, operators deploying secure analytics, organizations whose network structure may be analyzed, and the research community.

\noindent\textbf{Principles.}
We consider the Menlo Report principles of Beneficence, Respect for Persons, Justice, and Respect for Law and Public Interest.

\noindent\textbf{Potential harms.}
A key potential harm is the unintended disclosure of sensitive relational information. Even aggregate statistics such as cycle counts can, in some settings, reveal information about individuals or organizations. Additionally, our system intentionally reveals noisy degree information to achieve efficiency. While this leakage is formally bounded under edge local differential privacy, it nonetheless represents a relaxation of the strongest possible secrecy guarantees.

Another potential concern is dual use. Efficient privacy preserving graph analytics could be deployed in settings where participants may not fully understand or consent to the analysis of their relational data, even if the computation itself is cryptographically secure.

\noindent\textbf{Mitigations.}
We mitigate these risks in several ways. First, all raw graph data remains secret shared and is never revealed to any computing party. Second, the only intentional leakage is through formally defined differentially private noisy degrees, with parameters chosen to provide strong privacy guarantees against inference of individual edges. Third, our protocol is designed under a malicious adversary model, reducing the risk that a deviating party can extract additional information beyond the prescribed output.

We do not collect real world user data in our experiments. All evaluations are performed on synthetic graphs or publicly available datasets that do not contain personally identifiable information.

\noindent\textbf{Decision.}
We believe it is ethical to proceed and publish because the work enables useful graph statistics without centralizing sensitive edges, and because the efficiency leakage is explicit and formally bounded. Publishing also enables scrutiny and responsible reuse, and reduces incentives to deploy less principled approaches with weaker privacy protections.

\bibliographystyle{IEEEtran}
\bibliography{output}

\appendix
\newpage
\section{Appendix}
\subsection{Protocols}
\begin{mybox2}{\algotitle{alg:wfcheck}{$\Pi_{\Proc{WFCheck}}$\\(Protocol for checking well-formedness)}}
\textbf{Parameters:} $n,t$ \\
\textbf{Input:} Secret-shared neighbor lists $\langle \mathcal{L}\rangle=\{\langle l_i\rangle\}_{i\in[n]}$ and public padded lengths $\{\hat d_i\}_{i\in[n]}$ \\
\textbf{Output:} $1/0$ \\
Each server $\Ser_b$:
\begin{enumerate}[nosep, label=\arabic*.]
\item Let $k=\lceil \log_{10} n\rceil$ and $B=10^k$. \Note{Base for canonical edge encoding}

\item \textbf{For} $i \in [n]$:
\begin{enumerate}[nosep, label=\alph*.]
\item Securely verify that the entries in $\langle l_i\rangle$ are strictly increasing.
      \Note{Sortedness}
\item Securely count the number of dummy entries in $\langle l_i\rangle$
      (i.e., values $\Cmp(\langle l_i\rangle, n{+}1)$) and check that this count is at most $2t$.
      \Note{Bounded dummy padding}
\item \textbf{If} either check fails, \textbf{return} $0$.
\end{enumerate}
\item \textbf{EndFor}

\item Sample a random placeholder value $\langle p\rangle$. \Note{Used to collapse all dummy edges}
\item Initialize an empty secret-shared list $\langle \Q\rangle$. \Note{Multiset of canonical edge encodings}

\item \textbf{For} $i \in [n]$:
\item \hspace{0.25cm} \textbf{For} $j \in [\hat d_i]$:
\begin{enumerate}[nosep, label=\roman*.]
\item Let $\langle v\rangle \gets \langle l_i[j]\rangle$.
\item Securely compute $\langle a\rangle \gets \min(i,\langle v\rangle)$ and $\langle b\rangle \gets \max(i,\langle v\rangle)$.
\item Securely compute $\langle g\rangle \gets \Mult(\langle a\rangle,B) + \langle b\rangle$.
      \Note{Canonical encoding}
\item Securely set $\langle g\rangle \gets \Mult(\Cmp(\langle v\rangle,n),\langle g\rangle) + \Mult(1-\Cmp(\langle v\rangle,n),\langle p\rangle)$.
      \NoteNL{Dummy neighbors ($v>n$) map to placeholder $p$; real neighbors keep $g$}
\item Append $\langle g\rangle$ to $\langle \Q\rangle$.
\end{enumerate}
\item \hspace{0.25cm} \textbf{EndFor}
\item \textbf{EndFor}

\item Let $e \gets \sum_{i\in[n]} \hat d_i$ (so $|\Q|=e$).
\item \textbf{If} $\Equal(e \bmod 2, 1)$, append $\langle p\rangle$ to $\langle \Q\rangle$.
      \Note{Ensure even length for pairing}

\item Securely shuffle $\langle \Q\rangle$, then securely sort $\langle \Q\rangle$.

\item \textbf{For} $r \in [|\Q|/2]$:
\item \hspace{0.25cm} Securely check whether $\Equal(\langle \Q[2r-1]\rangle,\langle \Q[2r]\rangle)$.
\item \hspace{0.25cm} \textbf{If} not equal, \textbf{return} $0$.
\item \textbf{EndFor}

\item \textbf{Return} $1$.
\end{enumerate}
\end{mybox2}
\newpage

\begin{mybox2}{\algotitle{alg:initstate}{$\Pi_\Proc{InitializeState}$}}
\textbf{Input:} Outer-loop structure $\langle \To \rangle$, padding length $\hat{d}$, vertices $n$, flag $\Op{isReset}\in\{0,1\}$ \\
\textbf{Output:} Secret-shared $\langle \To \rangle$, $\langle \Seen \rangle$, and $\langle \US\rangle$ \\
Each server $\Ser_b$:
\begin{enumerate}[nosep, label=\arabic*.]
\item \textbf{If} $\Equal(\Op{isReset},0)$: \NoteNL{First call: pad the freshly-set, still-unpadded $\langle \To\rangle$}
\item \hspace{0.25cm} \textbf{For} $i \in [n]$:
\item \hspace{0.5cm} Pad $\langle \To[i]\rangle$ to length $\hat{d}$ using dummy values from $[(n+1),n+\hat{d}]$
\item \hspace{0.25cm} \textbf{EndFor}
\item \hspace{0.25cm} Append $\hat{d}$ dummy lists to $\langle \To\rangle$, where each dummy list is of the form
      $(\langle n\!+\!r\rangle,\; \langle n\!+\!r\rangle,\ldots,\langle n\!+\!r\rangle)$ for $r\in[\hat d]$.
\item \textbf{Else} \NoteNL{Reset: $\langle \To\rangle$ is already padded and complete from setup, reuse it unchanged}
\item \textbf{EndIf}
\item $\langle \mathsf{C}\rangle \gets \Fshuffle(\Op{Copy}(\langle \To\rangle))$ \NoteNL{$\C$ ranges over all (real and dummy) vertex IDs}
\item Initialize $\langle D \rangle \gets \langle \mathsf{C}[0][:] \rangle$ \Note{directory of unseen vertex IDs}
\item $(\langle \Tu \rangle, \langle H_u \rangle) \gets \Pi_\Proc{HashTable}(\langle \mathsf{C} \rangle)$ \Note{build the unseen hash table over $\C$}
\item Initialize $\langle \Op{Pointer}\rangle[\langle D[q]\rangle] \gets q$ for $q\in[n]$ \NoteNL{$\Op{Pointer}$: vertex ID $\to$ its position in $D$, for $O(1)$ deletions}
\item $\langle \US \rangle \gets (\langle \Tu\rangle,\; \langle D \rangle,\; \Fshuffle(\langle D\rangle),\; \cnt_u=n,\;\Op{Pointer})$
\item $\langle \Seen \rangle \gets (\langle \Ts\rangle=\varnothing,\; \langle D_s\rangle=[\,],\; \cnt_s=0)$ \Note{Seen data structure starts empty}
\item \Return $(\langle \To \rangle, \langle \Seen \rangle, \langle \US\rangle)$
\end{enumerate}
\end{mybox2}

\newpage
\begin{mybox2}{\algotitle{alg:update-state}{$\Pi_{\Proc{UpdateState}}$}}
\textbf{Input:}
Seen data structure $\langle \Seen \rangle$,
unseen data structure $\langle \US \rangle$,
unseen access list $\langle \Ru \rangle$, and the current outer-loop vertex $i$
\\
\textbf{Output:}
Updated $\langle \Seen \rangle$ and $\langle \US \rangle$
\\[0.3em]
Each server $\Ser_b$:
\begin{enumerate}[nosep, label=\arabic*.]
\item Let $\langle V_u\rangle \gets \{\langle v_u\rangle : (\langle r_u\rangle,\langle v_u\rangle)\in \langle \Ru\rangle\}$ \NoteNL{Discard the real/fake tags; both are promoted to seen}
\item \textbf{Update seen hash table and directory.}
\item \hspace{0.25cm} $\langle D_s \rangle.\Op{Append}(\langle V_u \rangle \cup \{i\})$
\NoteNL{Add vertices in $\Ru\cup\{i\}$ to the seen directory}
\item \hspace{0.25cm} $\langle \mathsf{C}_s \rangle \gets \Op{Copy}(\langle D_s \rangle)$
\item \hspace{0.25cm} $(\langle \T_s \rangle,\; \langle \mathsf{H}_{s} \rangle)
      \gets \Pi_\Proc{HashTable}(\langle \mathsf{C}_s \rangle)$
\item \hspace{0.25cm} Set $\langle \Ts \rangle \gets (\langle \T_s \rangle,\; \langle \mathsf{H}_{s} \rangle)$
\item \hspace{0.25cm} Set $\langle \cnt_s \rangle \gets \langle \cnt_s \rangle + \hat d_i + 1$

\item \textbf{Update unseen directory.}
\item \hspace{0.25cm} \textbf{For} each $\langle v \rangle \in \langle V_u \rangle \cup \{i\}$:
\item \hspace{0.5cm} $\langle p \rangle \gets \Op{Pointer}[\langle v \rangle]$ \Note{Pointer: current position of $v$ in $D$}
\item \hspace{0.5cm} $\langle w \rangle \gets \langle D[\cnt_u] \rangle$
\item \hspace{0.5cm}  Swap $\langle D[p] \rangle$ and $\langle D[\cnt_u] \rangle$
\item \hspace{0.5cm}  $\Op{Pointer}[\langle w \rangle] \gets \langle p \rangle$
\item \hspace{0.5cm}  $\langle \cnt_u \rangle \gets \langle \cnt_u \rangle - 1$
\item \hspace{0.25cm} \textbf{EndFor}

\item \textbf{Refresh unseen sampling view.}
\item \hspace{0.25cm} Set $\langle D' \rangle \gets \Fshuffle(\langle D[1..\cnt_u] \rangle)$

\item \Return $(\langle \Seen \rangle,\; \langle \US \rangle)$
\end{enumerate}
\end{mybox2}

\newpage
\begin{mybox2}{\algotitle{alg:gen-seen}{$\Pi_{\Proc{GenSeenAccesses}}$}}
\textbf{Input:} $\langle l \rangle$-secret-shared neighbor list;  $\la \Seen \ra$ - seen data structure (for membership test and sampling), $d$ - Number of access requests
\\
\textbf{Output:} $\la \Rs \ra$ - Secret-shared list of $d$ seen-access requests, each of the form $\Rs[j]=(\update{r_s},v_s)$ with $\update{r_s\in\{0,1\}}$ indicating a \emph{real} seen access \\
Each server $\Ser_b$
\begin{enumerate}[nosep, label=\arabic*.]
\item Parse $\Seen$ as $(\langle \Ts\rangle,\langle D_s \rangle,\cnt_s)$; only $D_s$ (seen vertex IDs) and $\cnt_s$ are used here.
\item Securely sample $d$ random \emph{distinct} indices from $[\cnt_s]$;
      denote this secret-shared set of indices by $\langle \mathsf{I}\rangle$
\item Initialize secret-shared table $\la C \ra$ of length $d$,
      with every entry set to $(0, n{+}1)$ \Note{$C[j] = (\Op{SeenBit}_j, \Op{TrueVal}_j)$}
\item \textbf{For} each vertex $\la v\ra \in \la l \ra$
  \item \hspace{0.52cm} Securely compute
          $\langle \Op{SeenBit}_j\rangle
             = \Op{IsInList}(\langle l[j]\rangle, \langle D_s\rangle)$
\NoteNL{Checks if the vertex $l[j]$ is already in the Seen list}
    \item \hspace{0.25cm} Set $\la C[j] \ra = (\langle \Op{SeenBit}_j\rangle, \langle l[j]\rangle)$
    \item \textbf{EndFor}
\item Initialize a secret-shared table $\langle C' \rangle$ of length $d$,
      with every entry set to $(0, n{+}1)$
      \Note{$C'[k] = (\Op{Valid}_k, \Op{FakeVal}_k)$}

\item \textbf{For } $k \in [d]$:
 \item \hspace{0.25cm} Set $\langle \Op{FakeVal}_k \rangle = \langle D_s[\mathsf{I}[k]] \rangle$ \NoteNL{Dereference the sampled index into a candidate seen vertex}
 \item \hspace{0.25cm} Set
          $\langle \Op{Valid}_k \rangle = 1 - \Op{IsInList}(\langle \Op{FakeVal}_k \rangle, \la l \ra)$
          \NoteNL{A fake candidate is invalid if it collides with a real neighbor in $l$}
    \item \hspace{0.25cm} Set $\langle C'[k] \rangle = (\langle \Op{Valid}_k \rangle, \langle \Op{FakeVal}_k\rangle)$.
\item \textbf{EndFor}
\item Securely sort $C$ in descending order by the first component (the seen-bit)
\NoteNL{Real seen vertices first; unseen later.}
\item Securely sort $C'$ in descending order by the first component (the validity-bit) \NoteNL{Valid fake candidates first; invalid later}
\item Initialize a secret-shared table $\langle \Rs \rangle$ of length $d$.

\item \textbf{For } $j \in [d]$:
    \item \hspace{0.25cm} Securely compute $\Rs[j]= \big(\Op{SeenBit}_j,\; \Op{SeenBit}_j?  \Op{TrueVal}_j : \Op{FakeVal}_j\big)$
\NoteNL{Real seen entries keep their true vertex and are tagged $r_s=1$; others use a random fake candidate tagged $r_s=0$.}
          \item \textbf{EndFor}
          \item Securely shuffle $ \la \Rs \ra$ \Note{Shuffle preserves each $(r_s,v_s)$ pairing}
          \item \textbf{Return} $ \la \Rs \ra$
\end{enumerate}
\end{mybox2}

\begin{mybox2}{\algotitle{alg:gen-unseen}{$\Pi_{\Proc{GenUnseenAccesses}}$}}
\textbf{Input:} $\langle l \rangle$ secret shared neighbor list;
$\langle \Seen \rangle$ seen data structure (for membership test);
$\langle \US \rangle$ unseen data structure (for sampling); $d$ number of access requests
\\
\textbf{Output:} $\langle \Ru \rangle$ secret shared list of $d$ unseen access requests, each of the form $\Ru[j]=(\update{r_u},v_u)$ with $\update{r_u\in\{0,1\}}$ indicating a \emph{real} unseen access
\\[0.2em]
Each server $\Ser_b$:
\begin{enumerate}[nosep, label=\arabic*.]
\item Parse $\langle \Seen \rangle$ as $(\langle \Ts\rangle,\langle D_s \rangle,\cnt_s)$; only $D_s$ (seen vertex IDs) is used here.
\item Parse $\langle \US \rangle$ as $(\langle \Tu \rangle,\langle D \rangle,\langle D_u \rangle,\cnt_u,\Op{Pointer})$; only $D_u$ is used here.
\item Securely sample $d$ random \emph{distinct} indices from $\langle D_u \rangle$; denote them by $\langle \mathsf{I}\rangle$.

\item Initialize secret shared table $\langle C \rangle$ of length $d$ with entries $(0,n{+}1)$
\Note{$C[j] = (\Op{UnseenBit}_j,\Op{TrueVal}_j)$}.
\item \textbf{For} $j \in [d]$:
\item \hspace{0.25cm} Securely compute $\langle \Op{SeenBit}_j \rangle \gets \Op{IsInList}(\langle l[j]\rangle,\langle D_s\rangle)$.
\item \hspace{0.25cm} Set $\langle \Op{UnseenBit}_j \rangle \gets 1 - \langle \Op{SeenBit}_j \rangle$.
\item \hspace{0.25cm} Set $\langle C[j]\rangle \gets (\langle \Op{UnseenBit}_j\rangle,\langle l[j]\rangle)$.
\item \textbf{EndFor}

\item Initialize secret shared table $\langle C' \rangle$ of length $d$ with entries $(0,n{+}1)$
\Note{$C'[k] = (\Op{Valid}_k,\Op{FakeVal}_k)$}.
\item \textbf{For} $k \in [d]$:
\item \hspace{0.25cm} Set $\langle \Op{FakeVal}_k \rangle \gets \langle D_u[\mathsf{I}[k]] \rangle$.
\item \hspace{0.25cm} Securely compute $\langle \Op{Collide}_k \rangle \gets \Op{IsInList}(\langle \Op{FakeVal}_k\rangle,\langle l\rangle)$.
\item \hspace{0.25cm} Set $\langle \Op{Valid}_k \rangle \gets 1 - \langle \Op{Collide}_k \rangle$
\Note{invalid if it duplicates a real neighbor}.
\item \hspace{0.25cm} Set $\langle C'[k]\rangle \gets (\langle \Op{Valid}_k\rangle,\langle \Op{FakeVal}_k\rangle)$.
\item \textbf{EndFor}

\item Securely sort $\langle C \rangle$ in descending order by the first component
\Note{real unseen neighbors first}.
\item Securely sort $\langle C' \rangle$ in descending order by the first component
\Note{valid fake unseen first}.

\item Initialize a secret shared table $\langle \Ru \rangle$ of length $d$.
\item \textbf{For} $j \in [d]$:
\item \hspace{0.25cm} Securely compute
$\langle \Ru[j]\rangle \gets \big(\Op{UnseenBit}_j,\; \Op{UnseenBit}_j \;?\; \Op{TrueVal}_j \;:\; \Op{FakeVal}_j\big)$.
\item \textbf{EndFor}
\item Securely shuffle $\langle \Ru \rangle$. \Note{Shuffle preserves each $(r_u,v_u)$ pairing}
\item \Return $\langle \Ru \rangle$.
\end{enumerate}
\end{mybox2}

\newpage
\begin{mybox2}{\algotitle{alg:hashtable}{$\Pi_{\Proc{HashTable}}$}}
\textbf{Parameters:} $p$ (public field modulus for $H(\cdot)$); $x$ (public oversampling factor, chosen large enough to keep collision probability negligible) \\
\textbf{Input:} List $\langle L\rangle$ of length $d$.\\
\textbf{Output:} Hash table $\langle T\rangle$ of size $m$, hash key $\langle H(\cdot)\rangle$.
\\
\begin{enumerate}[nosep, label=\arabic*.]
\item $\langle L'\rangle \gets \Op{Shuffle}(\langle L\rangle)$
\item $m \gets x\cdot d$ \hfill \Note{Table size, oversampled by factor $x$}
\item Securely sample $\langle a\rangle,\langle b\rangle$ uniformly at random in [p]
\item Define $H(\langle y \rangle) \gets \Op{Mod}(\Mult(\langle a \rangle,\langle y \rangle) + \langle b \rangle, p)$
\item Initialize $\langle T[0..m-1]\rangle \gets \emptyset$
\item \textbf{For} $t=0$ \textbf{to} $d-1$:
\begin{enumerate}[nosep, label=\alph*.]
\item $\Op{Insert}(\langle T\rangle, H, \langle L'[t]\rangle)$
\end{enumerate}
\item Return $\langle T\rangle, \langle H(\cdot)\rangle$ \NoteNL{Use $\langle T\rangle$ for any lookups over $\langle L'\rangle$ via the same $H(\cdot)$}.
\end{enumerate}
\end{mybox2}

\begin{mybox2}{\algotitle{alg:psi-enum}{$\Pi_{\Proc{PSI-Enum}}$}}
\textbf{Input:}
two secret-shared padded neighbor lists $\langle l_a\rangle$ and $\langle l_b\rangle$,
each of length $d$ (with dummy IDs in $[n{+}1,\,n{+}d]$),
and an access-validity flag $\langle f\rangle\in\{0,1\}$ from $\Pi_{\Proc{EnumerateTriangles}}$
\\
\textbf{Output:}$\langle \mathcal{X}\rangle$ containing the common neighbors, count of triangles $\langle count \rangle$
\\[0.2em]
Each server $\Ser_b$:
\begin{enumerate}[nosep, label=\arabic*.]
\item Initialize $\langle \mathcal{X}\rangle \gets [\,]$ and $\langle count \rangle \gets 0$.

\item \textbf{For} $t \in [d]$:
\begin{enumerate}[nosep,label=\alph*.]
\item Let $\langle x\rangle \gets \langle l_a[t]\rangle$.
\item \textbf{For} $u \in [d]$:
\begin{enumerate}[nosep,label=\roman*.]
\item $\langle y\rangle \gets \langle l_b[u]\rangle$
\item $\langle eq \rangle \gets \Equal(\langle x\rangle,\langle y\rangle)$
\item $\langle dum \rangle \gets \Cmp(\langle x\rangle,n)$
\item $\langle k \rangle \gets \Mult(\langle f\rangle,\langle eq \rangle,\langle dum \rangle)$
\item $\langle \mathcal{X} \rangle$.Append($\Mult(\langle x\rangle,\langle k \rangle)+\Mult((1-\langle k \rangle),\langle 0\rangle)$)
\item $\langle count \rangle \gets \langle count \rangle + \langle k \rangle$
\end{enumerate}
\item \textbf{EndFor}
\end{enumerate}
\item \textbf{EndFor}
\item Return $\langle \mathcal{X} \rangle, \langle count \rangle$
\end{enumerate}
\end{mybox2}

\newpage

\begin{mybox2}{\algotitle{alg:enum-tri}{$\Pi_{\Proc{EnumerateTriangles}}$}}
\textbf{Input:}
round index $i$,\\ $\langle l_i\rangle$ neighbor list of vertex $i$ of length $\hat d_i$,\\
seen-request list $\langle \Rs\rangle$,\\
unseen-request list $\langle \Ru\rangle$,\\
seen lookup structure $\langle \Ts\rangle$,\\
unseen lookup structure $\langle \Tu\rangle$\\
\textbf{Output:}
Triangle tuples contributed in round $i$: $\langle \Li\rangle$, and count of triangles $\langle count \rangle$
\\[0.2em]
Each server $\Ser_b$:
\begin{enumerate}[nosep, label=\arabic*.]
\item Initialize $\langle \Li\rangle \gets [\,], \langle count \rangle \gets 0$

\item \textbf{For} $j \in [\hat d_i]$:
\begin{enumerate}[nosep,label=\alph*.]

\item \textbf{Fetch accessed vertices.}
\begin{enumerate}[nosep,label=\roman*.]
\item $(\langle r_s\rangle,\langle v_s\rangle) \gets \langle \Rs[j]\rangle$ ; $\langle r_s\rangle\in\{0,1\}$ indicates a \emph{real} seen access.
\item $(\langle r_u\rangle,\langle v_u\rangle) \gets \langle \Ru[j]\rangle$; $\langle r_u\rangle\in\{0,1\}$ indicates a \emph{real} unseen access.
\end{enumerate}

\item \textbf{Lookup neighbor lists.}
\begin{enumerate}[nosep,label=\roman*.]
\item $\langle l_{v_s}\rangle \gets \Op{Lookup}(\langle \Ts\rangle,\langle v_s\rangle)$
\item $\langle l_{v_u}\rangle \gets \Op{Lookup}(\langle \Tu\rangle,\langle v_u\rangle)$
\end{enumerate}

\item \textbf{Compute validity masks.}\NoteNL{checks real or dummy, prevents reverse counting, and suppresses fake accesses}
\begin{enumerate}[nosep,label=\roman*.]
\item $\langle f_s\rangle \gets \Mult(\Cmp(\langle v_s\rangle,n),\\\Cmp(\langle v_s\rangle,i),\langle r_s\rangle)$
\item $\langle f_u\rangle \gets \Mult(\Cmp(\langle v_u\rangle,n),\\\Cmp(\langle v_u\rangle,i),\langle r_u\rangle)$
\end{enumerate}

\item \textbf{Compute intersection and append triangles.}
\begin{enumerate}[nosep,label=\roman*.]
\item $\langle \mathcal{X}_s\rangle, \langle c \rangle \gets \Pi_{\Proc{PSI-Enum}}(\langle l_i\rangle,\langle l_{v_s}\rangle,\langle f_s \rangle$)
\item $\langle count \rangle \gets \langle count \rangle+\langle c \rangle$
\item$\langle \Li\rangle$.Append($i, \langle v_s \rangle, \langle \mathcal{X}_s\rangle$)

\item $\langle \mathcal{X}_u\rangle, \langle c \rangle \gets \Pi_{\Proc{PSI-Enum}}(\langle l_i\rangle,\langle l_{v_u}\rangle,\langle f_u \rangle$)
\item $\langle count \rangle \gets \langle count \rangle+\langle c \rangle$
\item $\langle \Li\rangle$.Append($i, \langle v_u \rangle, \langle \mathcal{X}_u\rangle$)
\end{enumerate}

\end{enumerate}
\item \textbf{EndFor}
\item \Return $(\langle count\rangle,\langle \Li\rangle)$ 
\end{enumerate}
\end{mybox2}

\newpage

\begin{mybox2}{\algotitle{alg:sum-run-bin}{$\Pi_{\Proc{SumRunBinomials}}$: protocol for run aggregation over the sorted list $W'$}}
\textbf{Input:} Sorted list $\langle W' \rangle$ of length $L = 2\hat{d}_i^2$. \\
\textbf{Output:} Quadrangle count, List of quadrangles 
\begin{enumerate}[nosep, label=\arabic*.]
\item Initialize $\langle count \rangle \gets 0, \langle run \rangle \gets 1, \langle \Li \rangle \gets [\,], \langle \mathit{buf} \rangle \gets [\,]$
\NoteNL{$run$: length of the current maximal run of equal $val$'s; $\mathit{buf}$: the $v$'s seen in that run}
\item $\langle (p_{val}, p_{v}) \rangle \gets \langle W'[0] \rangle$ \Note{$p_{val},p_v$: the previous row}
\item $\langle \mathit{buf} \rangle.\text{Append}(\langle p_{v} \rangle)$

\item \textbf{For} $t = 1$ \textbf{to} $L-1$:
\begin{enumerate}[nosep,label=\Alph*.]
    \item $\langle (curr_{val}, curr_{v}) \rangle \gets \langle W'[t] \rangle$ \Note{the current row}
    \item $\langle same \rangle \gets \Equal(\langle curr_{val} \rangle,\langle p_{val} \rangle)$ \Note{current row continues the run}
    \item $\langle flag \rangle \gets \Mult((1 - \langle same \rangle),\Cmp(\langle p_{val} \rangle,n))$ \NoteNL{run just ended on a real (non-dummy) $val$: flush it}
    \item $\langle count \rangle \gets \langle count \rangle + \Mult(\langle flag \rangle,\binom{\langle run \rangle}{2})$
    \item $\langle \Li \rangle.\text{Append}(\Mult(\langle flag \rangle,(\langle p_{val} \rangle, \langle \mathit{buf} \rangle)) + \Mult((1 - \langle flag \rangle),\langle 0 \rangle))$
    \NoteNL{keep the compressed representation $(val,\mathit{buf})$ without materializing the $\binom{|buf|}{2}$ quadrangle tuples.}
    \item $\langle \mathit{buf} \rangle \gets \Mult(\langle same \rangle,\langle \mathit{buf} \rangle) + \Mult((1 - \langle same \rangle),[\,])$
    \item $\langle \mathit{buf} \rangle.\text{Append}(\langle curr_{v} \rangle)$
    \item $\langle run \rangle \gets \Mult(\langle same \rangle,(\langle run \rangle + 1)) + \Mult((1 - \langle same \rangle),1)$
    \item $\langle p_{val} \rangle \gets \langle curr_{val} \rangle$
\end{enumerate}

\item \textbf{Finalize:} 
    \begin{enumerate}[nosep,label=\roman*.]
    \item $\langle f_{final} \rangle \gets \Cmp(\langle p_{val} \rangle,n)$
    \item $\langle \Li \rangle.\text{Append}(\Mult(\langle f_{final} \rangle,(\langle p_{val} \rangle, \langle \mathit{buf} \rangle)) + \Mult((1 - \langle f_{final} \rangle),\langle 0 \rangle))$
    \NoteNL{stored in the same compressed form.}
    \item $\langle count \rangle \gets \langle count \rangle + \Mult(\langle f_{final} \rangle,\binom{\langle run \rangle}{2})$
    \end{enumerate}
\item \Return $(\langle count\rangle, \langle \Li\rangle)$
\end{enumerate}
\end{mybox2}
\newpage

\begin{mybox2}{\algotitle{alg:enum-quad}{$\Pi_{\Proc{EnumerateQuadrangles}}$}}
\textbf{Input:}
round index $i$,\\ $\langle l_i\rangle$ neighbor list of vertex $i$ of length $\hat d_i$,\\
seen-request list $\langle \Rs\rangle$,\\
unseen-request list $\langle \Ru\rangle$,\\
seen lookup structure $\langle \Ts\rangle$,\\
unseen lookup structure $\langle \Tu\rangle$\\
\textbf{Output:}
Quadrangle tuples contributed in round $i$: $\langle \Li\rangle$, and count of quadrangles $\langle count \rangle$
\\[0.2em]
Each server $\Ser_b$:
\begin{enumerate}[nosep, label=\arabic*.]
\item Initialize $\langle \Li\rangle \gets [\,], \langle count \rangle \gets 0$
\item Initialize $\langle W \rangle \gets [\,]$

\item \textbf{For} $j \in [\hat d_i]$:
\begin{enumerate}[nosep,label=\Alph*.]

\item \textbf{Fetch accessed vertices.}
\begin{enumerate}[nosep,label=\roman*.]
\item $(\langle r_s\rangle,\langle v_s\rangle) \gets \langle \Rs[j]\rangle$ ; $\langle r_s\rangle\in\{0,1\}$ indicates a \emph{real} seen access.
\item $(\langle r_u\rangle,\langle v_u\rangle) \gets \langle \Ru[j]\rangle$; $\langle r_u\rangle\in\{0,1\}$ indicates a \emph{real} unseen access.
\end{enumerate}

\item \textbf{Lookup neighbor lists.}
\begin{enumerate}[nosep,label=\roman*.]
\item $\langle l_{v_s}\rangle \gets \Op{Lookup}(\langle \Ts\rangle,\langle v_s\rangle)$
\item $\langle l_{v_u}\rangle \gets \Op{Lookup}(\langle \Tu\rangle,\langle v_u\rangle)$
\end{enumerate}

\item \textbf{Compute validity masks.}\NoteNL{checks real or dummy, prevents multiple counting, and suppresses fake accesses}
\begin{enumerate}[nosep,label=\roman*.]
\item $\langle f_s\rangle \gets \Mult(\Cmp(\langle v_s\rangle,n),\\\Cmp(\langle v_s\rangle,i),\langle r_s\rangle)$
\item $\langle f_u\rangle \gets \Mult(\Cmp(\langle v_u\rangle,n),\\\Cmp(\langle v_u\rangle,i),\langle r_u\rangle)$
\end{enumerate}

\item \textbf{Compute 2-hop lengths and build W.}
{\small
\begin{enumerate}[nosep,label=\roman*.]
\item \textbf{For} $k \in [\hat d_i]$:\NoteNL{Represent each valid 2-hop path compactly as $(val,v)$ rather than materializing pairs of paths into quadrangles.}
\begin{enumerate}[nosep,label=\alph*.]
\item $\langle f_{sk}\rangle \gets \Mult(\Cmp(\langle l_{v_s}[k]\rangle,n),\Cmp(\langle l_{v_s}[k]\rangle,i))$
\item $\langle f_{uk}\rangle \gets \Mult(\Cmp(\langle l_{v_u}[k]\rangle,n),\Cmp(\langle l_{v_u}[k]\rangle,i))$
\item $\langle b_s \rangle \gets \Mult(\langle f_{sk}\rangle,\langle f_s\rangle)$
\item $\langle b_u \rangle \gets \Mult(\langle f_{uk}\rangle,\langle f_u\rangle)$
\item $\langle val \rangle \gets \Mult(\langle b_s \rangle,\langle l_{v_s}[k]\rangle) + \Mult((1-\langle b_s \rangle),(\langle n+1 \rangle))$
\item $\langle w \rangle \gets (\langle val \rangle, \langle v_s \rangle)$
\item $\langle W \rangle$.Append($w$)
\item $\langle val \rangle \gets \Mult(\langle b_u \rangle,\langle l_{v_u}[k]\rangle) + \Mult((1-\langle b_u \rangle),(\langle n+1 \rangle))$
\item $\langle w \rangle \gets (\langle val \rangle, \langle v_u \rangle)$
\item $\langle W \rangle$.Append($w$)
\end{enumerate}
\end{enumerate}}
\end{enumerate}
\item \textbf{EndFor}
\item $\langle W' \rangle \gets$ Sort($\langle W \rangle$) by first item \NoteNL{Sorting groups paths with the same endpoint $val$, allowing each group to remain compressed as $(val,\mathit{buf})$.}
\item $(\langle count\rangle,\langle \Li\rangle) \gets \Pi_{\Proc{SumRunBinomials}}(\langle W' \rangle)$ \NoteNL{Count quadrangles from run sizes while retaining the compressed representation for enumeration.}
\item \Return $(\langle count\rangle,\langle \Li\rangle)$ 
\end{enumerate}
\end{mybox2}

\subsection{Synthetic Dataset Generation}
\label{app:synthetic_dataset}
We generate the two synthetic datasets using an Erd\H{o}s--R\'enyi $G(n,m)$ model \cite{bollobas2011random} with an explicit maximum-degree cap. The Erd\H{o}s--R\'enyi $G(n,m)$ model generates a random undirected graph by selecting exactly $m$ edges uniformly at random from all possible $\binom{n}{2}$ vertex pairs. Given $n$ and target density $\rho$, we set the target number of edges to
$m=\mathrm{round}\!\left(\rho\cdot \frac{n(n-1)}{2}\right)$.
We then sample candidate endpoints $u,v$ uniformly at random from $\{1,\ldots,n\}$ (rejecting self-loops), and add the undirected edge $\{u,v\}$ if it has not appeared before. If a max-degree bound $d_{\max}$ is provided, we reject any candidate edge incident to a vertex whose current degree is already $d_{\max}$.
We repeat until exactly $m$ edges are accepted; if the target is infeasible under the cap (i.e., $m > \lfloor n\cdot d_{\max}/2\rfloor$) or cannot be reached within a fixed attempt budget, the generator aborts.
All synthetic graphs are generated with a fixed random seed for reproducibility.

\subsection{Cheaper Non-Linear Operations using FSS}\label{app:FSS}
\textbf{Function Secret Share (FSS).}
A two-party function secret sharing~\cite{boyle2015function,boyle2016function} scheme enables splitting a function $f$ into succinct function shares such that individual shares provide no information about function $f$, yet when their evaluations are aggregated at any input point $x$, they reconstruct the correct output $f(x)$. A two-party FSS construction consists of the following pair of algorithms:
\begin{itemize}
    \item $K_1, K_2 \leftarrow \text{Gen}(1^\lambda, f)$: Takes as input the security parameter $1^\lambda$ and a function specification $f$, generates and returns two keys $K_1$ and $K_2$.
    \item $y_i \leftarrow \text{Eval}(K_i, x)$: Takes as input a key $K_i$ and an evaluation point $x$, then outputs a value $y_i$ representing the $i$-th party's share of $f(x)$. 
\end{itemize}
The correctness property ensures that combining the evaluation outputs $y_1 + y_2 = f(x)$ for any input $x$.
We focus on distributed point functions (DPFs), which implement FSS for point functions $f_{\alpha,\beta}$ where $f_{\alpha,\beta}(\alpha) = \beta$ and $f_{\alpha,\beta}(\alpha') = 0$ for all $\alpha' \neq \alpha$. We denote the DPF generator algorithm as $\text{Gen}(1^\lambda, \alpha, \beta)$, with evaluation performed using $\text{Eval}$.
For $f_{\alpha,\beta}$,  if $e = \text{Eval}(K_1, x) + (K_2, x)$ for x over the entire function domain $|\B|$, the expanded shares are of the one-hot form: $\mathbf{e} = (0, \ldots, 0, \beta, 0, \ldots, 0)$ 
where  $\beta$ appears at position $\alpha$. The evaluated share for party $i$ at $x$ is then the $i$-th share of this one-hot encoding, i.e., $y_i(x)$ such that $y_1(x) + y_2(x) = f_{\alpha, \beta}(x)$ for all $x$. We use EvalAll to denote a full domain
evaluations which can be performed more efficiently.
\\\\\noindent\textit{Verifiable Function Secret Sharing (VFSS).}
VFSS extends FSS by enabling parties to check that the shares are well-formed. Specifically, in the context of DPFs, it allows the parties to check that the expanded vector is indeed  a valid one-hot encoding. This is typically done with the help of malicious sketching protocols~\cite{wagh2022pika,boneh2021lightweight}.
\\\\\noindent\textbf{Offline Phase.} Let $\B$ denote the domain of the FSS. In the offline phase, we leverage the dealer to generate the FSS keys. The dealer runs $\mathcal{F}_{\textsf{Rand}}$ with $\Ser_0$ and $\Ser_1$ to get $r_0, r_1 \in \mathbb{F}$.  Next, it generates FSS keys for the DPF $f_{r',1}$ where $r'=(r_0+r_1)\bmod \B$. These keys can later be converted to correspond to the DPFs of the actual bids during the online phase. 
We now describe how \sysname{} securely verifies the correctness of these keys under the malicious threat model. Our design follows the general approach proposed by~\cite{wagh2022pika}, which we briefly summarize here for convenience. 
In this setting, two main challenges must be addressed:
\begin{itemize}
\item  \textit{$\Ser_b$ is malicious.} In this case, although the FSS keys are correctly generated by the dealer, the shares obtained by expanding these keys will later be converted into secret shares of the bid encodings during the online phase. Hence, the challenge is to ensure that the servers perform the correct computation over these expanded shares when executing the auction algorithm in the online phase. To address this, we require the expanded values to be authenticated secret-shares, i.e., shares augmented with SPDZ-style MACs (see Section~\ref{sec:background}). Since the DPF encodes a point function, the corresponding MAC shares also form a DPF. As a result, we can meet this requirement by having the dealer generate not only the FSS keys for the DPF $f_{r',1}$, but also additional FSS keys for the MAC shares $f_{r',\alpha}$, where $\alpha$ is a global MAC key.
\item \textit{Dealer is malicious.}  Here, we have to verify that the DPF keys and correlated randomness $r$ generated by the dealer are well-formed. This corresponds to verifiable FSS  and is enforced with the help of a malicious sketching scheme. 
\end{itemize} 
Concretely, in the offline phase the dealer samples  $r \in \mathbb{F}$ and a MAC key $\alpha \in \mathbb{F}$ and generates two pairs of DPF keys as:
\[
(\K_{r'}^0,\K_{r'}^1)\randarrow \text{FSS.gen}(1^\lambda, r', 1)
\]
\[
(\K_{\alpha}^\Br,\K_{\alpha}^\Ex)\randarrow \text{FSS.gen}(1^\lambda, r', \alpha)
\]
\[
r'=r\bmod \B
\]

The dealer then sends to the server $\Ser_b$ the keys $\{\K_{r}^b, \K_{\alpha}^b\}$, along with the linear secret shares $\{[r]_b,[\alpha]_b\}$. 
Next, servers engage in a malicious sketching protocol which verifies the following properties:
\begin{itemize}[itemsep=1pt]
\item $\K^{b}_{r'}, \K^{b}_{\alpha}$ form shares of well formed DPFs
\item $[r]_b$ form shares of  the index of the DPFs (modulo $\B$)
\item The expanded shares $[y]_{b}=\textsf{FSS.evalAll}(\K^b_{r'})$ and $[m_y]_{b}=\textsf{FSS.evalAll}(\K^b_{\alpha})$ satisfy $m_y=\alpha\cdot y$
\end{itemize}

The details of the sketching scheme are provided in below. Instead of using two separate keys, an optimization is encoding both the function and its MAC shares within a single DPF key by extending its payload. This would effectively halve the cost of \textsf{FSS.evalAll}. 
\subsubsection{Sketching Scheme}
\label{app:sketch}

To ensure malicious security against the ill-formed FSS keys, we employ a sketching scheme to verify the correctness of these keys. Over a DPF lookup of size $d$, the servers evaluate vectors $y$ and $m_y$, of length $d$. The scheme has two primary objectives: verification should always succeed for well-formed keys, and malformed vectors $y$ or $m_y$ should be accepted with only negligible probability. The input shares are constructed over a finite field $\mathbb{F}$.
We employ the following linear sketch matrix $L \in \mathbb{Z}_{F}^{4 \times d}$:

\[L = \begin{pmatrix}
a_{1,1} & a_{1,2} & \cdots & a_{1,d} \\
a_{2,1} & a_{2,2} & \cdots & a_{2,d} \\
a_{3,1} & a_{3,2} & \cdots & a_{3,d} \\
a_{4,1} & a_{4,2} & \cdots & a_{4,d}
\label{eq:sketch}\end{pmatrix}\]

For each $i \in [d]$, the entries $a_{1,i}, a_{2,i}$ are sampled uniformly from $\mathbb{F}$ (in implementation, servers share PRG seeds to generate these values), $a_{3,i} = a_{1,i} \cdot a_{2,i} \pmod{\mathbb{F}}$, and $a_{4,i} = i-1$.
The intuition behind this construction is clarified in Eq. \ref{eq:sketching_zero}. This matrix requires generation only once and can be reused throughout the protocol unless a party aborts. Let $L_1, L_2, L_3, L_4$ represent the rows of $L$.

\textbf{Verification.} Define $z_j = \langle L_j, y \rangle \in \mathbb{F}$ for $j = 1, 2, 3, 4$. Furthermore, let $z^* = \langle L_1, m_y \rangle \in F$. Note that each party only holds a share of these values, i.e., $y$ (and thus each $z_j$) is secret shared over $F$ between the servers.
The parties then evaluate the following expression (within MPC) and check if the value (as an element of $\mathbb{F}$) is equal to 0: 
\begin{align}
(z_1 z_2 - z_3) + (z_4 - r\bmod \B) &\stackrel{?}{=} 0    \label{eq:verification_a}\\
(\alpha z_1 - z^*) &\stackrel{?}{=} 0\label{eq:verification_b}
\end{align}
The browser and exchange run this verification over MPC and use this verification to check correct/incorrect keys.

\textbf{Completeness.} The completeness of this sketching scheme is pretty intuitive. For honestly generated correlated randomness, $(z_4 - r\bmod \B)$ by the construction of the DPF scheme equals zero (the DPF outputs 1 at $r'$ where $r'=r\bmod \B$). At the same time, since $y$ is of hamming weight one, the first term is zero:
\begin{equation}
\label{eq:sketching_zero}
    z_1 z_2 - z_3 = \sum_i a_{1,i} a_{2,i} (y_i^2 - y_i) + \sum_{i \neq j} a_{1,i} a_{2,j} (y_i y_j) = 0
\end{equation}

, where the first term is zero because each $y_i \in \{0,1\}$ and the second term is zero because if at most one of $y_i, y_j$ is non-zero for any pair of distinct $i, j$. Thus each term equals zero for Eq.~\ref{eq:verification_a}. Moreover, the verification Eq.~\ref{eq:verification_b} is valid by honest construction.

\textbf{Correctness.}
The correctness of the protocol is easy to follow from the correctness of the DPF:
\[u[i] = y[i+x] \neq 0 \text{ iff } i+x = r'\]

\subsection{Data Structures}\label{app:DS}

\noindent\textbf{Edge-retirement flags: examples.}
Here, we give examples of the function of the edge-retirement flags used in our protocols. Recall that after degree sorting and relabeling, vertices are assigned public IDs $1,\ldots,n$ (Figure~\ref{fig:degree}).

\smallskip
\noindent\emph{Triangles.}
Let $\N(1)=\{2,3,4\}$ and $\N(2)=\{1,4\}$.
When processing $i=1$, our protocol considers only neighbors $j\in\N(1)$ with $j>1$.
For $j=2$, we find $\N(1)\cap \N(2)=\{4\}$ and output $(1,2,4)$.
Later, when processing $i=2$, the edge $(2,1)$ is ignored since it violates the flag $j>i$.
Thus, the same triangle is not rediscovered via the backward orientation $(2,1,4)$ (although it can still be found via another forward edge, e.g., $(1,4)$).

\smallskip
\noindent\emph{Quadrangles.}
For quadrangle detection, the edge-retirement flag is applied to \emph{every} vertex touched during a round, not just to the first edge.
That is, throughout the entire round, the algorithm only allows visiting vertices with larger IDs than the current one.

Consider a 4-cycle on vertices $1,2,3,4$, with edges $(1,2)$, $(2,3)$, $(3,4)$, and $(4,1)$.
When the algorithm processes vertex $1$, it can visit both $2$ and $4$, since they have larger IDs.

From each of these, it can then visit $3$, which is also larger than $1$.
Because $3$ is reached from two different paths, the algorithm recognizes this as a quadrangle and outputs the cycle $(1,2,3,4)$.

When the algorithm later processes vertex $2$, discovering the same cycle would require visiting vertex $1$ again.
This is disallowed because the ordering rule forbids touching any vertex with a smaller ID than the one currently being processed.
The same restriction applies when processing vertices $3$ and $4$.
As a result, the quadrangle is produced only when processing the smallest vertex on the cycle, ensuring that each quadrangle is counted exactly once.

\subsection{Complexity Analysis}
\label{app:complexity}

\noindent\textbf{Overview.}
The work of \sysname{} is governed by the public noisy degrees.
Let $n=|V|$, and let $\hat d_i$ denote the noisy degree of vertex
$\bar v_i$ under the decreasing order $\prec$. In this section, we derive \textit{total} MPC
work. Here, operations that are independent can be parallelized and therefore
reduce wall-clock time, which is discussed in the following section. For secure sorting, we use
the shuffle--then--sort paradigm and model sorting $m$ secret-shared
elements as $O(m\log m)$ work.

Across \sysname{}'s three phases: data collection, setup, and loop
execution, the dominant costs arise during loop execution, which
consists of state maintenance, request generation, and cycle detection.
Which component dominates depends on the graph's degree distribution.
Below, we analyze the contribution of each routine to the overall
complexity.

\noindent\textbf{Data collection.}
Each party locally perturbs the degrees of its vertices and pads each
neighbor list to its published noisy degree before secret sharing.
Hence, the total number of submitted adjacency entries is
\[
    \Theta\!\left(
        \sum_{i=1}^{n}\hat d_i
    \right).
\]

\noindent\textbf{One-time setup.}
The well-formedness checks scan every submitted adjacency entry once,
requiring
\[
    \Theta\!\left(
        \sum_{i=1}^{n}\hat d_i
    \right)
\]
work. For global edge consistency, the servers construct a list $Q$
containing one canonical encoding for every submitted adjacency entry,
so
\[
    |Q|
    =
    \Theta\!\left(
        \sum_{i=1}^{n}\hat d_i
    \right).
\]
The servers then obliviously shuffle and sort $Q$, followed by a
linear scan of the sorted list. The sort therefore dominates the
well-formedness check, giving
\[
    O\!\left(
        \left(\sum_{i=1}^{n}\hat d_i\right)
        \log\!\left(\sum_{i=1}^{n}\hat d_i\right)
    \right)
\]
setup work.

\noindent\textbf{State maintenance.}
In each round, $\hat d_i$ vertices are consumed from the unseen data structures
and promoted to the seen data structures. Within a reset period, after processing
some rounds, suppose $k$ vertices have been consumed in total. The seen
state then contains $k$ vertices, while the unseen directory contains
$n-k$ vertices. Rebuilding $\Ts$ therefore costs $O(k)$, while
shuffling the active unseen directory costs $O(n-k)$ under a
linear-shuffle cost model~\cite{cryptoeprint:2025/2137}. Their combined
cost in that round is
\[
    O(k)+O(n-k)=O(n).
\]
This holds independently of the noisy degrees and of where the round
falls within a reset period: larger seen-state rebuilds are accompanied
by proportionally smaller unseen-directory shuffles. Since the protocol
executes $n$ rounds, the total per-round state-maintenance cost is
therefore
\[
    \sum_{i=1}^{n} O(n)
    =
    O(n^2).
\]

The unseen hash table $\Tu$ is rebuilt only when the remaining unseen
pool is too small to serve the next round. Across all rounds, the total
number of vertices consumed from the unseen pool is
\[
    \sum_{i=1}^{n}\hat d_i.
\]
Each fresh unseen pool contains $n$ vertices, so the number of resets is
\[
    O\!\left(
        \frac{\sum_i\hat d_i}{n}
    \right).
\]
\footnote{A reset may leave some vertices unused, but decreasing
noisy-degree order ensures that each pool is at least half consumed
before resetting. Thus, the exact number of reset periods is
$\Theta(1+\frac{1}{n}\sum_i\hat d_i)$.}
Since rebuilding $\Tu$ processes all $n$ vertices, the total reset cost
under a linear-shuffle model is
\[
    O\!\left(
        n\cdot\frac{\sum_i\hat d_i}{n}
    \right)
    =
    O\!\left(
        \sum_i\hat d_i
    \right).
\]
Since each $\hat d_i=O(n)$, we have
$\sum_i\hat d_i=O(n^2)$, so the reset cost is subsumed by the
$O(n^2)$ per-round state updates. Therefore, with a linear shuffle,
state maintenance requires
\[
    O(n^2)
\]
total work.

\noindent\textbf{Per-round request generation.}
In round $i$, \sysname{} generates exactly $\hat d_i$ seen requests
and $\hat d_i$ unseen requests. The main cost is determining which
neighbors of the current vertex have already been seen. Before round
$i$ can proceed, at least $\hat d_i$ vertices must remain in the unseen
pool; otherwise the state is reset. Hence, at most $n-\hat d_i$
vertices can be in the seen pool. Testing each of the $\hat d_i$
neighbors against this pool therefore requires
\[
    O\!\left(\hat d_i(n-\hat d_i)\right)
\]
FSS comparisons in round $i$. Summed over all rounds, the membership
work is
\[
    O\!\left(
        n\sum_i\hat d_i-\sum_i\hat d_i^2
    \right).
\]

The remaining collision checks contribute
$O(\hat d_i^2)$ work per round, or
$O(\sum_i\hat d_i^2)$ across all rounds. Combining these costs gives
\[
    O\!\left(
        n\sum_i\hat d_i-\sum_i\hat d_i^2
        +\sum_i\hat d_i^2
    \right)
    =
    O\!\left(
        n\sum_i\hat d_i
    \right)
\]
total request-generation work. The FSS comparisons are independent
and parallelized in our implementation, significantly reducing their
wall-clock cost.
\\\noindent\textbf{Triangle enumeration.}
In round $i$, the servers retrieve $2\hat d_i$ neighbor lists, one
seen and one unseen request for each of the $\hat d_i$ access
positions. For each accessed vertex, triangle enumeration intersects its neighbor
list with $l_i$. Both lists have active length $\Theta(\hat d_i)$, so
the MPC-friendly pairwise intersection performs
$\Theta(\hat d_i^2)$ equality checks. Across the
$2\hat d_i$ accessed lists, round $i$ therefore requires
\[
    \Theta\!\left(
        2\hat d_i\cdot\hat d_i^2
    \right)
    =
    \Theta(\hat d_i^3)
\]
work. Summing over all rounds gives
\[
    \boxed{
    T_{\triangle}
    =
    \Theta\!\left(
        \sum_{i=1}^{n}\hat d_i^3
    \right)
    }
\]
for triangle enumeration.

\noindent\textbf{Quadrangle enumeration.}
In round $i$, the protocol processes the two-hop paths induced by the
$2\hat d_i$ retrieved neighbor lists, each of active length
$\hat d_i$. This produces a secret-shared list $W$ of fixed size
\[
    2\hat d_i^2.
\]
The protocol securely sorts $W$ by endpoint and then linearly scans the
sorted list to aggregate equal runs. The scan requires
$\Theta(\hat d_i^2)$ work, while sorting requires
\[
    \Theta\!\left(
        \hat d_i^2\log\hat d_i
    \right),
\]
which dominates the scan. Hence, over all rounds, quadrangle enumeration
requires
\[
    \boxed{
    T_{\square}
    =
    \Theta\!\left(
        \sum_{i=1}^{n}
        \hat d_i^2\log\hat d_i
    \right).
    }
\]

\noindent\textbf{Overall complexity.}
Combining setup, state maintenance, request generation, and cycle
detection, triangle enumeration requires
\[
    O\!\left(
        \left(\sum_{i=1}^{n}\hat d_i\right)
        \log\!\left(\sum_{i=1}^{n}\hat d_i\right)
        +
        n^2
        +
        n\sum_{i=1}^{n}\hat d_i
        +
        \sum_{i=1}^{n}\hat d_i^3
    \right).
\]
Since $\sum_i\hat d_i=O(n^2)$, the setup term is subsumed by
$n\sum_i\hat d_i$. Thus,
\[
    \boxed{
    O\!\left(
        n^2
        +
        n\sum_{i=1}^{n}\hat d_i
        +
        \sum_{i=1}^{n}\hat d_i^3
    \right).
    }
\]
Similarly, quadrangle enumeration requires
\[
    \boxed{
    O\!\left(
        n^2
        +
        n\sum_{i=1}^{n}\hat d_i
        +
        \sum_{i=1}^{n}
        \hat d_i^2\log\hat d_i
    \right).
    }
\]
Thus, the dominant component depends directly on the noisy-degree
distribution: request generation depends on the total noisy degree,
while cycle enumeration depends on higher-order degree terms.

For sparse graphs, where $\sum_i\hat d_i=O(n)$, the request-generation
cost $O(n\sum_i\hat d_i)$ reduces to $O(n^2)$ and is subsumed by the
state-maintenance term. Hence, triangle and quadrangle enumeration
require, respectively,
\[
    O\!\left(
        n^2+\sum_i\hat d_i^3
    \right)
    \quad\text{and}\quad
    O\!\left(
        n^2+\sum_i\hat d_i^2\log\hat d_i
    \right).
\]

\noindent\textbf{Comparison with dense-graph complexity.}
The bounds above show that \sysname{} benefits
directly from graph sparsity. For intuition, consider a graph where
$\hat d_i$ is approximated by the average noisy degree $\bar{\hat d}$. The triangle
complexity then becomes
\[
    O(n^2+n^2\bar{\hat d}+n\bar{\hat d}^{\,3})
    =
    O(n^2\bar{\hat d}+n\bar{\hat d}^{\,3}).
\]
This is better than the $O(n^3)$ dense baseline as long as
$\bar{\hat d}<n^{2/3}$. This covers a broad range of realistic sparse
graphs, where the average degree is typically much smaller than the
number of vertices~\cite{graphbenchmark}. For example,
constant average degree gives $O(n^2)$ complexity, while
$\bar{\hat d}=\sqrt n$ gives $O(n^{5/2})$.

Similarly, under the same approximation, quadrangle enumeration
requires
\[
    O(n^2\bar{\hat d}
      +n\bar{\hat d}^{\,2}\log\bar{\hat d}).
\]
Since $\bar{\hat d}=O(n)$, even at the worst density, \sysname{} costs at most
$O(n^3\log n)$, remaining asymptotically below the $O(n^4)$ dense
baseline. Thus, \sysname{} improves over dense-graph computation across
typical sparse regimes for triangles and across the full degree range
for quadrangles. We also highlight below that
\sysname{} is highly parallelizable.

\subsection{Parallelizations}\label{app:parallel} To improve efficiency, \sysname{} incorporates parallelization wherever it does
not affect security or correctness. Many subroutines operate on different
vertices (or different neighbor pairs) independently once the required
neighbor lists have been retrieved, making them natural targets for
data-parallel execution.

\noindent\textit{(i) Detection.}
Recall from Sec.~\ref{sec:protocol:round} that in round $i$ the servers first
retrieve $2\hat d_i$ neighbor lists (real and fake) needed for processing
vertex~$i$, and then compute the triangle (or quadrangle) contributions
involving~$i$. The key observation is that the \emph{post-retrieval detection
computation} for each round depends only on the lists retrieved in that round.
Therefore, \sysname{} decouples execution into two stages: (1) the oblivious retrievals and recording the $2\hat d_i$ retrieved
lists for each $i$, and (2) computing triangle and quadrangle contributions
from these recorded lists. The latter stage is embarrassingly parallel across
rounds because each vertex $i$ can be processed independently.

Concretely, for triangle counting and enumeration, the work for computing contributions for a fixed
vertex $i$ consists of computing neighborhood intersections between $l_i$ and
each accessed neighbor list. Since lists are padded to length $\hat d_i$, this
requires $\Theta(2\hat d_i \cdot \hat d_i^2)$ pairwise equality checks (two sets of $\hat d_i$ neighbor lists, each intersected with a
length-$\hat d_i$ list). These comparisons are independent and can be further
parallelized within a round, in addition to parallelizing across different
vertices.

For quadrangle counting and enumeration (Sec.~\ref{sec:quadrangle}), the
computation for a fixed vertex $i$ first scans the $2\hat d_i$ retrieved lists
to generate two hop path information. Each two hop path
contributes an entry to the list $W$, which is then securely sorted and scanned
to compute run lengths and add the corresponding $\binom{\ell}{2}$ terms.  These steps are still
independent across vertices $i$, so we can parallelize across rounds, but within a
single round the main coordination points are the run-length aggregation over
the sorted $W$ (which depends on contiguous equal blocks).

In practice, parallelism is bounded by the available CPU cores. We therefore
treat each round $i$ as a job with weight proportional to its expected
comparison count, and schedule jobs across cores using a simple greedy
load-balancing strategy: we sort rounds by decreasing weight and assign each
round to the currently least-loaded worker. This yields near-uniform utilization
in practice and allows \sysname{} to exploit multi-core hardware while
preserving the protocol's oblivious access pattern.

\noindent \textit{(ii) Computing hash.}
Hash evaluations are independent across items, so \sysname{} parallelizes hash
computations across vertices and across elements within each neighbor list.

\noindent \textit{(iii) Oblivious sorting.}
Oblivious sorting is comparison-based; \sysname{} parallelizes the underlying
comparison operations across cores.

\noindent \textit{(iv) Generating requests.}
Seen/unseen request generation consists of sampling indices and computing
membership and collision flags, all of which are independent per element and
therefore parallelized.

\update{\subsection{Ideal Functionalities $\mathcal{F}_{\textsf{Sort}}$ and $\mathcal{F}_{\textsf{Shuffle}}$}\label{app:hybrid-func}
The proof below is stated in the $(\mathcal{F}_{\textsf{Sort}},\mathcal{F}_{\textsf{Shuffle}},\mathcal{F}_{\textsf{Rand}})$-hybrid model. We take $\mathcal{F}_{\textsf{Rand}}$ (an ideal source of shared randomness) as standard; we define the other two below.
\begin{packeditemize}
\item $\mathcal{F}_{\textsf{Sort}}$: on input a secret-shared list $\langle L\rangle$ and a secret-shared key for each entry, outputs a secret-shared list $\langle L'\rangle$ containing the same entries as $\langle L\rangle$, permuted so that the keys are in descending order. Both the input list and the output list are secret-shared throughout; only the fact that a sort occurred, not the resulting permutation, is revealed to the servers.
\item $\mathcal{F}_{\textsf{Shuffle}}$: on input a secret-shared list $\langle L\rangle$ of length $m$, samples a uniformly random permutation $\pi$ of $[m]$ and outputs the secret-shared list $\langle L'\rangle$ with $L'[k]=L[\pi(k)]$ for all $k$, revealing neither $\pi$ nor any information about it to either server.
\end{packeditemize}
Secure realizations of both exist in the literature (e.g., the shuffle protocol of~\cite{cryptoeprint:2025/2137} used in our implementation, Sec.~\ref{sec:evaluation}); by the composition theorem, replacing calls to these protocols with calls to the ideal functionalities does not affect the overall security of \sysname.}

\subsection{Proof of Theorem \ref{thm:security}}\label{app:security}
\begin{proof}
We begin by first providing a construction for our
simulator $\textsf{Sim}$ for the ideal world. We work in the hybrid
model, where invocations of sub-protocols can be replaced
with that of the corresponding functionalities, as long as the
sub-protocol is proven to be secure. We will operate in the
$(\mathcal{F}_{\textsf{Sort}},\mathcal{F}_{\textsf{Shuffle}})$-hybrid model \update{(App.~\ref{app:hybrid-func})}, where we assume the existence of a
secure protocol
which realizes the desired ideal functionalities  $\mathcal{F}_{\textsf{Sort}}$ and $\mathcal{F}_{\textsf{Shuffle}}$.
For the ease of exposition, we start with the case of honest data owning parties. WLOG let us assume $\Ser_1$ is the corrupt server. The simulator sends  random values as the shares of the honest parties to the $\Ser_1$. All the secret-shared messages from $\Ser_0$ to $\Ser_1$ can be simulated using dummy values. Now, let's look at round 1 and see what are the additional messages in the view of $\Ser_1$. In particular, $\Ser_1$ sees $\hat{d}_1$ indices that are accessed in $\Tu$ and corresponding $\hat{d}_1$ indices accessed in $D$ to delete these entries. All of these indices are unique and recall we use a \update{secret-shared hash function $h(x)=ax +b mod p$ (Sec.~\ref{sec:protocol:setup})}. The hash table itself is constructed after secure shuffling the collection of neighbor lists. Hence, the simulator $\textsf{Sim}$ can just sample $\hat{d}_1$ indices at random from $[p]$, where $p$ is the output domain of the hash function and $[\textsf{cnt}_u]$ to simulate them for $\Ser_!$. \update{Since each server sees only its own additive share of the hash function's key $(a,b)$, these simulated distributions are information-theoretically (not merely computationally) indistinguishable from $\Ser_1$'s real view.} From round 2 onwards, in addition to the above indices, $\Ser_1$ also sees $\hat{d}_i$ indices for sampling the candidate seen accesses, indices for accessing $\Ts$ and indices for sampling candidate unseen accesses from $D_u$. Note that all of these indices are unique by our construction. Again, both $D_u$, the list of seen vertices in $\Seen$  are securely shuffled in every round and $\Ts$ is also rebuilt every time. Hence, these indices can also be simulated by choosing $\hat{d}_i$ uniformly random values from $[\textsf{cnt}_u]$,$[\textsf{cnt}_s]$ and [p] respectively.
The final value of either the cycle count or enumerated cycles can be simulated via access to the ideal functionality  to construct consistent values.
The malicious data owners do not receive any message from our protocol.
During the offline phase,   $\textsf{Sim}$ invokes $\mathcal{F}_{\textsf{Rand}}$ to get random values for $[\alpha]_0,[\alpha]_1,[r]_0,[r]_1$. $\textsf{Sim}$ generates the DPF keys honestly and sends a share of $[\alpha r]$ and $\K_r^1, \K_\alpha^1$ to $\calA$.  Using them, $\textsf{Sim}$ proceeds to complete the malicious sketch with $\calA$. If the check fails, $\textsf{Sim}$ sends
abort to the ideal functionality.
For handling a malicious dealer, 
$\textsf{Sim}$ invokes $\mathcal{F}_{\textsf{Rand}}$ to get random values for $\alpha_0,\alpha_1,r_0,r_1,c_0,c_1$. It receives shares of $\alpha r$ and the DPF keys from $\calA$. If the shares are invalid, $\textsf{Sim}$ sends abort to the ideal functionality. 
\end{proof}

\begin{table}
\centering

\begin{tabular}{l p{5cm}}
\toprule
\textbf{MPC Protocol} & \textbf{Output/Functionality}\\
\midrule
$\Mult(\langle x \rangle, \langle y \rangle)$ & $\langle x \cdot y\rangle$\\
$\textsf{Equal}(\langle x \rangle, \langle y \rangle)$ & $\langle 1\rangle$ if $x=y$, $\langle 0\rangle$ otherwise\\
$\Cmp(\langle x \rangle, \langle y \rangle)$ & $\langle 1\rangle$ if $x<y$, $\langle 0\rangle$ otherwise\\
$\textsf{Rand}(y)$ & $\langle r \rangle$ with a uniform random $y$-bit value $r$\\
$\textsf{Mod}(\langle x \rangle ,p)$ & $\langle x \bmod p \rangle$ \\
\bottomrule
\end{tabular}%
\caption{Basic MPC protocols~\cite{aliasgari2012secure} used in \sysname.}
\end{table}

\subsection{Related Work}
\label{app:related-work}
In \autoref{app:tab:related-work}, we compare prior work along task, output, threat model, leakage, DP role, and asymptotic cost.

\begin{table*}[t]
\centering
\caption{Comparison of approaches for private graph analytics and triangle counting. SH denotes semi-honest, Mal denotes malicious threat model.}
\label{app:tab:related-work}
\small
\resizebox{\textwidth}{!}{
\begin{tabular}{
p{0.32\textwidth}
p{0.09\textwidth}
p{0.09\textwidth}
p{0.20\textwidth}
p{0.12\textwidth}
p{0.18\textwidth}
}
\toprule
\textbf{Task / Approach} 
& \textbf{Output} 
& \textbf{Threat} 
& \textbf{Leakage (beyond output)} 
& \textbf{DP Role} 
& \textbf{Asymptotic Cost} \\
\midrule

\textbf{Iterative aggregation (message passing)} \\
\quad \textsc{GraphSC} \cite{nayak2015graphsc}, \textsc{GraSP} \cite{kapoor2025mathsf}, \textsc{SecureGraph} \cite{araki2021secure}, 
\textsc{Graphiti} \cite{graphiti}, \textsc{RingSG} \cite{zou2025ringsg}, \textsc{CogNN} \cite{zou2024cognn}
& Aggregates 
& Mal / SH 
& None 
& None 
& Irrelevant \\
\midrule

\textbf{Iterative aggregation (message passing + DP)} \\
\quad \textsc{Mazloom'18} \cite{mazloom2018secure}, \textsc{Mazloom'20} \cite{mazloom2020secure}
& Aggregates 
& Mal / SH 
& DP-protected access patterns (noisy degrees) 
& Topology leakage 
& Irrelevant \\
\midrule

\textbf{Cycle detection (MPC)} \newline
\hspace*{1.2em}\textsc{Vorstermans}~\cite{vorstermans2023secure} (matrix) \newline
\hspace*{1.2em}\textsc{Oryx}~\cite{zhong2024oryx} (path-based)
& \mbox{} \newline Exact \newline Exact
& \mbox{} \newline SH \newline SH
& \mbox{} \newline None \newline Structural ($k$-hop paths)
& \mbox{} \newline None \newline None
& \mbox{} \newline $O(n^3)$ \newline $O(v \cdot n^{k-1})$ \\
\midrule

\textbf{Triangle counting via joins} \\
\quad \textsc{Shrinkwrap}$^\dagger$~\cite{shrinkwrap}
& Exact
& SH
& DP-protected intermediate sizes
& Efficiency
& $O\!\left(n^2d^3\log(n^2d^3)\right)$ \\
\midrule

\textbf{DP triangle counting} \\
\quad \textsc{CARGO} \cite{liu2024cargo}, \textsc{MAGO} \cite{wang2023mago}
& Noisy 
& SH 
& DP-protected outputs (degrees, counts) 
& Output privacy 
& $O(n^3)$ \\
\midrule

\textbf{Triangle/quadrangle detection} \\
\quad \textbf{\sysname{}}
& \textbf{Exact$^\ast$} 
& \textbf{Mal} 
& \textbf{DP-protected degrees}
& \textbf{Efficiency}
& $O\big(
        n^2
        +
        n\sum_{i=1}^{n}\hat d_i$
        \\
& & & & &
        $+ \sum_{i=1}^{n}\hat d_i^3
    \big)$ \\
\bottomrule
\end{tabular}
}
\vspace{2pt}
{\footnotesize
$^\dagger$\textsc{Shrinkwrap} does not directly target triangle counting,
but can express it via join queries. Its reported complexity assumes a
$d$-regular graph and ignores DP-noise and secure-memory overhead; both
can further increase its cost.\\
$^\ast$\sysname{} can optionally release a noisy triangle count; however,
its exact triangle count release is fully secure and independent of this
(see Sec.~\ref{sec:DP} for details).}
\end{table*}

\subsection{SQL Join Formulation of Triangle and Quadrangle Counting.} \label{app:sql} Shrinkwrap~\cite{shrinkwrap} provides a natural SQL-based comparison to \sysname{}, since triangle and quadrangle counting can be expressed as sequences of joins over the edge relation. However, it is fundamentally less suitable for this setting for two reasons. First, Shrinkwrap provides security against semi-honest adversaries, whereas \sysname{} is maliciously secure. Second, although Shrinkwrap uses differential privacy (DP) to reduce intermediate cardinalities, this reduction occurs only \emph{after} each join has been securely evaluated with exhaustive padding~\cite{shrinkwrap}. Thus, consecutive graph joins still materialize increasingly large intermediate path relations, which must then be obliviously sorted before being resized. In contrast, \sysname{} uses DP-noisy degrees to bound the computation \emph{before} cycle detection, avoiding these global join intermediates altogether. We quantify this computational and privacy difference below.

\noindent\textbf{Triangle counting.} Triangle counting can be written as $\texttt{COUNT}(E(u,v)\Join E(v,w)\Join E(u,w))$. The first join pairs edges sharing $v$ and enumerates all length-2 paths, or wedges. Its true output size is $\Theta(\sum_i d_i^2)$, where $d_i$ is the degree of vertex $i$. However, the edge relation itself contains $\Theta(\sum_i d_i)$ tuples, so Shrinkwrap first evaluates this join into an exhaustively padded output of size $\Theta((\sum_i d_i)^2)$ and only afterward obliviously sorts and resizes that output to its DP-noisy cardinality~\cite{shrinkwrap}. In other words, even if the graph contains relatively few real wedges, the first secure join must initially accommodate every possible pair of edge tuples. The second join then takes the resized wedge relation and joins it again with the full edge relation to test whether each wedge is closed by an edge. Even ignoring DP noise and secure-memory overhead, this requires $\Theta((\sum_i d_i)(\sum_i d_i^2))$ candidate tuple processing. Shrinkwrap additionally sorts each exhaustively padded intermediate before resizing it, which costs $O(N\log N)$ for an intermediate of size $N$~\cite{shrinkwrap}. Thus, DP helps reduce the relation passed forward, but does not avoid the large padded join or its resize cost. A simple example illustrates the difference. In a $d$-regular graph,
$\sum_i d_i=\Theta(nd)$ and $\sum_i d_i^2=\Theta(nd^2)$.
Ignoring DP noise, the second SQL join is therefore exhaustively padded
to $\Theta(n^2d^3)$ tuples. Shrinkwrap must then obliviously sort this
entire padded intermediate before resizing it, adding
$O(n^2d^3\log(n^2d^3))$ work, even before accounting for secure-memory
overhead. With DP resizing, the intermediate can only be larger because
its input cardinality includes nonnegative DP noise. In contrast, \sysname{} avoids constructing a global wedge relation and requires $O(n^2d+nd^3)$ work under the same regular-degree approximation. The key difference is that the SQL formulation repeatedly multiplies a local path count by the size of the global edge relation, whereas \sysname{} keeps cycle detection local to each vertex. 

\noindent\textbf{DP sensitivity.} Shrinkwrap perturbs the output cardinality of each intermediate operator. Its noise therefore depends on the sensitivity of that operator's cardinality, which Shrinkwrap computes recursively from operator stability. In particular, the stability of a join depends on the maximum input multiplicity, so sensitivity can compound across successive joins~\cite{shrinkwrap}. Shrinkwrap's own example shows this effect explicitly: one join increases sensitivity from $1$ to $m$, while the following join increases it to $m^2$. For graph self-joins, these multiplicities are degree dependent; for example, the wedge-producing join has multiplicity bounded by $d_{\max}$. The fixed privacy budget must also be divided across the intermediate releases. In contrast, \sysname{} applies DP to vertex degrees before cycle detection. Under edge-DP, each degree has sensitivity $1$, so the noise in $\hat d_i=d_i+\eta_i$ does not scale with $d_{\max}$. These noisy degrees directly determine the local data sizes and loop bounds used throughout the protocol. 

\noindent\textbf{Quadrangle counting.} Quadrangle counting extends the same idea by one more join: $\texttt{COUNT}(E(u,v)\Join E(v,w)\Join E(w,x)\Join E(x,u))$. After constructing wedges, the next join extends them into length-3 paths, and the final join checks whether each path is closed. Each additional join therefore takes an already enlarged intermediate relation and combines it again with the full edge relation. Under Shrinkwrap, that join is again evaluated with exhaustive padding, followed by an oblivious sort and DP resize before the next operator can use the smaller result~\cite{shrinkwrap}. For intuition, consider again a $d$-regular graph. The number of length-3 paths is $\Theta(nd^3)$, so the final closing-edge join alone processes $\Theta(n^2d^4)$ candidate tuples even without DP noise. By comparison, \sysname{} performs quadrangle counting in $O(n^2d+nd^2\log d)$ under the same regular-degree approximation. Thus, the gap becomes larger for quadrangles: the SQL formulation repeatedly materializes and rejoins global path relations, while \sysname{} operates on degree-bounded local neighborhoods. Overall, Shrinkwrap reduces the padding propagated \emph{between} operators, but it does not remove the exhaustive computation needed to produce each operator's padded output. Its Resize procedure itself obliviously sorts that padded output before truncation~\cite{shrinkwrap}, and its RAM implementation additionally charges oblivious secure-array reads and writes for join processing. For graph cycle counting, these costs compound across consecutive joins. \sysname{} instead releases low-sensitivity noisy degrees before computation and uses them to organize localized cycle detection, avoiding the large global intermediates inherent to the SQL join plan.

\subsection{Extension to Other Tasks}\label{app:extension}
While this paper focuses on triangle and quadrangle detection, two of the most fundamental cycle-based primitives, our techniques extend naturally to a broader class of graph-analytic tasks.

At a high level, \sysname’s design is modular: the outer loop handles oblivious neighbor-list retrieval, while the inner loop performs a cycle-specific detection routine. To support additional patterns, one simply replaces the inner detection logic while retaining the same oblivious retrieval structure. 
More generally, many graph tasks involve vertices or edges annotated with attributes, for instance, timestamps, types, or weights, and cycles should be counted only when these attributes satisfy certain criteria. \sysname{} can support these extensions by incorporating attributes directly into neighbor lists. During the detection phase, the protocol can privately inspect these attributes and filter candidate cycles based on flags in the same way as we do currently.

\end{document}